\documentclass[10pt, final]{IEEEtran}
\usepackage{amsmath,amsfonts}
\usepackage{algorithmic}
\usepackage{algorithm}
\usepackage{array}
\usepackage{textcomp}
\usepackage{stfloats}
\usepackage{url}
\usepackage{verbatim}
\usepackage{graphicx}
\usepackage{cite}
\usepackage{float} 
\usepackage{authblk}
\usepackage{lipsum}
\usepackage[mathscr]{euscript}
\usepackage{framed}
\usepackage{bm}
\usepackage{booktabs}
\usepackage{stfloats}  
\usepackage{setspace}
\usepackage{amsthm} 

\ifCLASSOPTIONcompsoc
\usepackage[caption=false,font=normalsize,labelfont=sf,textfont=sf]{\part{title}subfig}
\else
\usepackage{subfigure}
\fi

\usepackage{graphicx,amsmath,amssymb,amsfonts}
\allowdisplaybreaks[3]
\usepackage{caption} 
\usepackage{hyperref}
\hypersetup{colorlinks=true}   
\usepackage{xcolor}
\usepackage[T1]{fontenc}   
\newtheorem{remark}{\textbf{Remark}}
\newtheorem{lemma}{\textbf{Lemma}}

\newtheorem{corollary}{\textbf{Corollary}}

\newtheorem{proposition}{\textbf{Proposition}}

\makeatletter
\renewcommand{\maketag@@@}[1]{\hbox{\m@th\normalsize\normalfont#1}}%
\makeatother

\makeatletter

\newcommand{\Rmnum}[1]{\expandafter\@slowromancap\romannumeral #1@}

\makeatother

\begin{document}
	\captionsetup{font={small}}
	\bstctlcite{reference:BSTcontrol}
	
	\title{\fontsize{18 pt}{\baselineskip}\selectfont Wireless Communication with Cross-Linked Rotatable Antenna Array: Architecture Design and Rotation Optimization}
	\author{
		\fontsize{10 pt}{\baselineskip}\selectfont Ailing~Zheng, Qingqing~Wu, Ziyuan~Zheng, Qiaoyan~Peng, Yanze~Zhu, Honghao~Wang, Wen~Chen, Guoying~Zhang
		\vspace{-10 mm}
		\thanks{A. Zheng, Q. Wu, Z. Zheng, Y. Zhu, H. Wang, W. Chen, and G. Zhang are with the Department of Electronic Engineering, Shanghai Jiao Tong University, Shanghai 200240, China (e-mail: (ailing.zheng, qingqingwu, zhengziyuan2024, yanzezhu, hhwang, wenchen, gy\_zhang)@sjtu.edu.cn).}
		\thanks{Q. Peng is with the Department of Electronic Engineering, Shanghai Jiao Tong University, Shanghai 200240, China, and also with the State Key Laboratory of Internet of Things for Smart City and the Department of Electrical and Computer Engineering, University of Macau, Macao SAR, China (email: qiaoyan.peng@connect.um.edu.mo).}
		
	}
	\vspace{-10 mm}
	
	\maketitle
	\vspace{-5 mm}
	\begin{abstract}
		Rotatable antenna (RA) technology can harness additional spatial degrees of freedom by enabling the dynamic three-dimensional orientation control of each antenna. Unfortunately, the hardware cost and control complexity of traditional RA systems is proportional to the number of RAs. To address the issue, we consider a cross-linked (CL) RA structure, which enables the coordinated rotation of multiple antennas, thereby offering a cost-effective solution. To evaluate the performance of the CL-RA array, we investigate a CL-RA-aided uplink system. Specifically, we first establish system models for both antenna element-level and antenna panel-level rotation. Then, we formulate a sum rate maximization problem by jointly optimizing the receive beamforming at the base station and the rotation angles. For the antenna element-level rotation, we derive the optimal solution of the CL-RA array under the single-user case.  
		Subsequently, for two rotation schemes, we propose an alternating optimization algorithm to solve the formulated problem in the multi-user case, where the receive beamforming and the antenna rotation angles are obtained by applying the minimum mean square error method and feasible direction method, respectively. In addition, considering the hardware limitations, we apply the genetic algorithm to address the discrete rotation angles selection problem. Simulation results show that by carefully designing the row-column partition scheme, the performance of the CL-RA architecture is quite close to that of the flexible antenna orientation scheme. Moreover, the CL antenna element-level scheme surpasses the CL antenna panel-level scheme by 25\% and delivers a 128\% performance improvement over conventional fixed-direction antennas.
	\end{abstract}
	
	
	\vspace{-5mm}
	\section{Introduction}
	The proliferation of intelligent devices and services is accelerating the circulation of massive data, pushing the frontiers of wireless networks. The forthcoming sixth generation (6G) network is envisioned to support an unprecedented density of connections across diverse application while achieving a leap in performance across throughput, latency, and reliability \cite{Wang2023}. This necessitates fundamental innovations to unlock additional degrees of freedom (DoFs) and achieve superior spectral efficiency without proportionally escalating hardware complexity and power consumption. To meet the  required substantial performance improvement, multiple-input multiple-output (MIMO) has been extensively investigated to enhance the transmission rate and reliability by exploiting spatial DoFs at the transceivers \cite{Bogale2016}. However, the performance gains offered by traditional MIMO technology are achieved at the expense of scaling up the number of antennas, which translates into significantly increased hardware costs, higher power consumption, and greater signal processing complexity \cite{Rusek2013}. Furthermore, simply increasing the number of antennas cannot fully exploit the available spatial DoFs of the MIMO since the positions and orientations of the antennas are fixed once deployed. The fixed deployment of antennas prevents the wireless system from fully exploiting the rich and continuous variations of practical wireless channels in the spatial domain. This inability to adaptively match the continuous channel structure ultimately constrains the achievable diversity and spatial multiplexing performance of MIMO systems \cite{Zhu2024CM}.
	
	To overcome the above limitations, movable antenna (MA) \cite{Zhu2024} and fluid antenna system (FAS) \cite{Wong2021} have been proposed to enable the local movement of antennas in a given region and attracted growing attention in wireless communication. Specifically, the authors of \cite{Zhu2024} developed a field-response model to characterize the general multi-path channel with given transmit and receive regions, where the maximum channel gain achieved by the MA was analyzed. In \cite{Wong2021}, the closed-form expression for the outage probability of a single-antenna FAS was derived over arbitrarily correlated Rayleigh fading channels.
	Generally, MA can flexibly adjust the antennas' positions based on channel state information, thus supporting high resolution and beamforming gains without increasing the number of antennas \cite{Zheng2025MA}. 
	By relocating antennas to favorable positions and reconfiguring the array geometry, the system enhances desired signal power, suppresses interference, and enables adaptive null steering. Such flexible rearrangement also reshapes the channel matrix to improve spatial multiplexing \cite{Ma2024,New2024}, thereby introducing additional spatial DoFs and intelligently reshaping the wireless channel to improve overall communication performance. 
	The unique capabilities of MAs have motivated their integration with diverse advanced applications, such as integrated sensing and communications
	(ISAC) \cite{Guo2024MA,Li2025,Wang2025MA-IRS}, intelligent reflecting surface (IRS) \cite{Zheng2025MA-3,Zheng2025MA-2,GaoMAIRS,Wu2025IRS}, unmanned aerial vehicle communications \cite{Liu2025,Tang2025,Tang2024UAV}, and non-orthogonal multiple access systems \cite{Xiao2025NOMA,Gao2025,Zhou2024NOMA}. 
	
	Unfortunately, despite its superior performance, MA suffers from limited response time and movement speed, which hinders its practical operation in fast-varying channels. Furthermore, conventional MAs only support positional adjustments while the orientation of antennas is fixed, which restricts their ability to fully track dynamic channel conditions. 
	To fully exploit spatial DoFs, six-dimensional (6D) MA architecture was proposed in \cite{Shao2025-1} and \cite{Shao2025}, which can dynamically adjust both its positions and rotations in response to varying channel conditions. This design enhances the network capacity, spatial multiplexing, and interference mitigation \cite{Shao2025-2,Ren2025}. As a lightweight implementation of 6D MA, rotation antenna (RA) with fixed positions has attracted increasing attention due to reduced hardware complexity and compact design \cite{Wu2025RA,Zheng2025-1,Zheng2025-2}. Specifically, RA arrays enable independent three-dimensional (3D) boresight control for each antenna, which in turn allows for precise beamforming toward desired directions and effective suppression of signals from undesired directions.

	Unlike MA that requires additional physical space for movement, RA only needs local rotational adjustment achieved by compact mechanical device. This makes RA more scalable and better compatible with existing wireless infrastructures. Various works have studied the RA-based communication systems in enhancing receive power \cite{Peng2025,Pan2025RA,Zhang2025RA}, securing communication \cite{Liang2025}, and assisting ISAC \cite{Zhou2025, Zheng2025MA-4} by jointly designing the transmit/receive beamforming and the RA boresight directions. 
	These studies have fully demonstrated the significant advantages brought by the RA structure.
	For practical RA deployment, various methods can be applied to realize antenna rotation, such as mechanically-driven (e.g., micro-electromechanical) \cite{Liang2025-1,Dai2025RA, Chang2003}, and electronically-driven (e.g., parasitic elements) \cite{Boyarsky2021,Zheng2025-2} RA. In general, mechanically-driven and electronically-driven methods serve as complementary alternatives for orientation control. The former excels in achieving a broader adjustment range, while the latter, in contrast, provides superior compatibility with current wireless platforms \cite{Zheng2025-3}. Additionally, depending on the requirement of the application, a trade-off between hardware cost and beamforming precision can be made by implementing rotation either at the antenna level or the panel level \cite{Tan2025RA, Wang2025MA}.
	
	Nevertheless, while existing research demonstrates the potential of RAs, it predominantly assumes that each antenna or panel requires an independent driver for full rotational freedom \cite{Zheng2025-2,Pan2025RA,Peng2025,Liang2025,Zhou2025,Zhang2025RA}.
	This approach leads to a critical drawback where hardware complexity, cost, and power consumption scale multiplicatively with the number of antenna elements, hindering practical deployment.
	To overcome this limitation, this paper considers a novel cross-linked (CL) RA architecture. The core idea is to mount antennas on shared rotation tracks along rows and columns, enabling coupled rotation at the row and column levels. This structural innovation reduces the number of required motors from a multiplicative order to an additive scale, dramatically lowering hardware costs and control complexity.
	However, this efficiency gain introduces a new and significant challenge. Specifically, the rotation angle of each antenna in the CL-RA design is no longer independent but is coupled by its row and column rotations. Consequently, steering the boresight of each antenna to precisely align with its intended user becomes a complex constrained problem. The previous boresight designs, which assume independent antenna control, are no longer applicable. 
	Therefore, it is imperative to unveil the performance potential of the CL-RA architecture and develop novel rotation schemes to fully unlock its capabilities under rotational coupling constraints.
	
	Motivated by the above issues, we investigates a CL-RA-enabled uplink system. The main contributions of this work are summarized as follows:
	\begin{itemize}
		\item Firstly, we present a CL-RA architecture and study an uplink multi-user communication system in this paper, where the base station (BS) is equipped with a CL-RA array to provide services to multiple single-antenna users. We establish the system model under antenna element-level rotation and formulate a sum rate maximization problem by jointly optimizing the receive beamforming and the rotation angles of all RAs. Then, we extend the framework to antenna panel-level rotation to further reduce the complexity and cost, presenting the corresponding system model and problem formulation. We also characterize the feasible range of rotation angles for the antenna panel-level rotation scheme.
		\item 
		Second, for the single-user case under antenna element-level rotation, we analytically derive the optimal receive beamforming and rotation angles, which reveals that for the uniform linear array (ULA) setup, the CL-RA architecture can achieve the same performance as the flexible design.
		For the multi-user case under both rotation schemes, we propose an alternating optimization (AO) algorithm. Specifically, the receive beamforming is optimized using the minimum mean square error (MMSE) method, while the rotation angles are updated via the feasible direction method within each AO iteration. Furthermore, to address the practical limitations in rotation flexibility, we extend our study to the discrete rotation angles selection problem, which is solved using a genetic algorithm.
		\item 
		Simulation results validate the effectiveness and convergence of the proposed algorithm. The CL-RA architecture exhibits a negligible performance loss relative to the ideal flexible antenna orientation scheme. Under practical physical constraints, the $2 \times 1$ panels partition scheme is 17\% better than that of the array-wise orientation design. Moreover, the performance of the CL antenna rotation scheme is superior to that of the CL panel rotation scheme, with an improvement of 25\%. Furthermore, the CL-RA architecture enjoys a 128\% higher performance gain than the conventional fixed orientation design. Additionally, the discrete rotation scheme attains an $84\%$ performance gain compared to the fixed-direction baseline.
	\end{itemize}
	
	The rest of this paper is organized as follows.
	Section \ref{System Model} first presents the CL-RA architecture. Then, the system model and the problem formulation for the antenna element- and panel-level rotation are provided. In Section \ref{Performance Analysis}, we analyze the optimal solution for the antenna element-level rotation under the single-user case. 
	Section \ref{Proposed Algorithm} proposes an AO algorithm to solve the formulated problem under the multi-user setup for two rotation schemes. Section \ref{Simulation Results} provides simulation results for performance evaluation and comparison. Finally, Section \ref{Conclusions} concludes this paper.
	
	\emph{Notations:} Vectors and matrices are denoted by boldface lowercase letters and boldface uppercase letters, respectively. $\mathbf{x}^{\mathrm{T}}$, $\mathbf{x}^{\mathrm{H}}$, $\mathbf{x}^{-1}$, and $\Vert \mathbf{x} \Vert$ denote the transpose, Hermitian, inverse, and 2-norm of vector $\mathbf{x}$, respectively. 
	$\mathbb{C}$ and $\mathbb{R}$ denote complex field and real field, respectively. $[x]^{+}=\max\{0,x\}$. $\mathbf{I}_M$ denotes the $M$-dimensional identity matrix. $\nabla$ denotes the differential operator.
	
	\vspace{-1mm}
	\section{System Model and Problem Formulation} \label{System Model}
	\subsection{CL-RA Architecture}
	\begin{figure}[t]
		\centering
		\includegraphics[width=0.43 \textwidth]{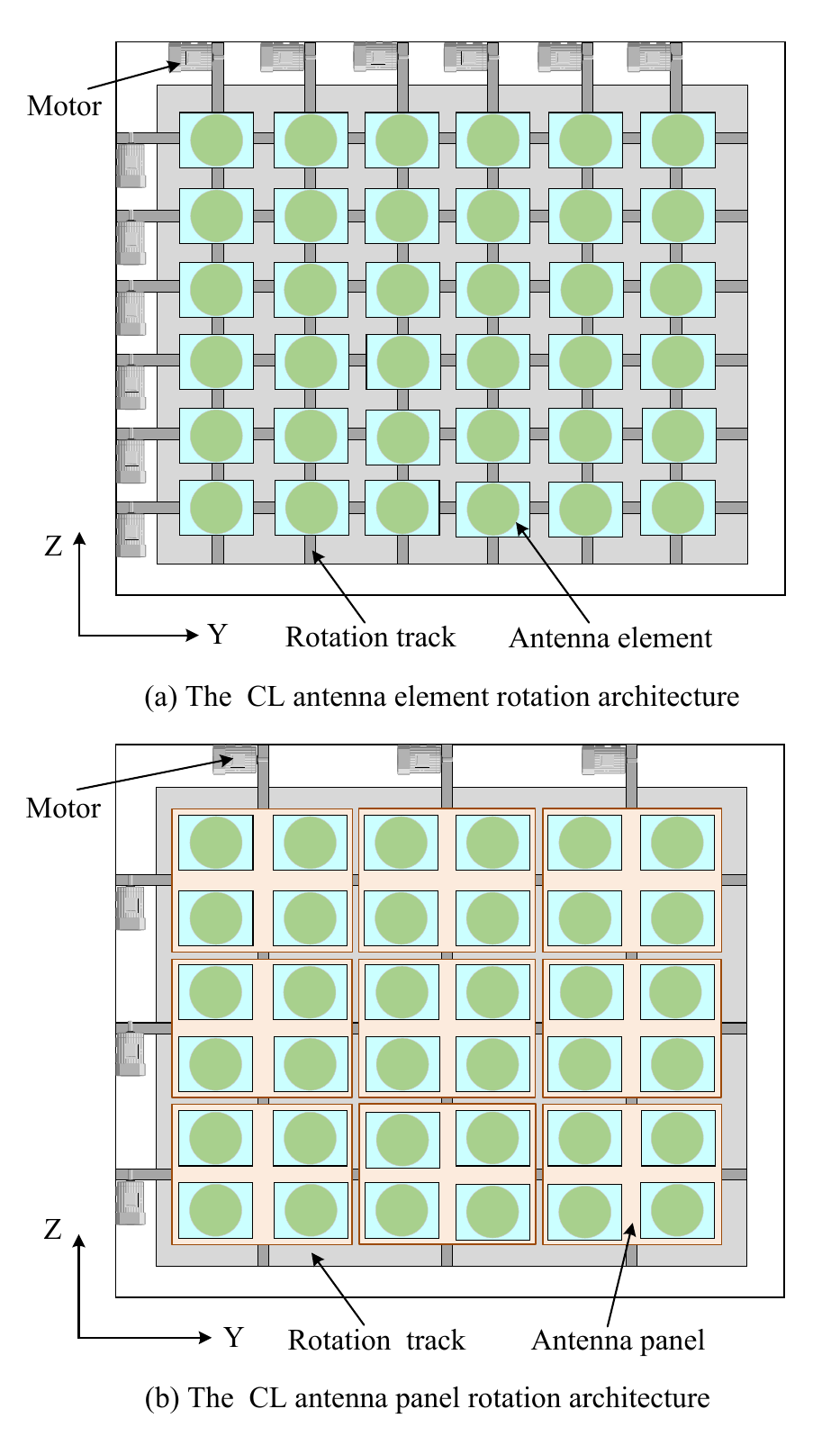}
		\caption{The considered CL-RA architecture.}
		\label{Fig1}
		\vspace{-5mm}
	\end{figure}
	We investigate a novel CL-RA array, as illustrated in Fig. \ref{Fig1}. The CL-RA array consists of multiple horizontal and vertical rotation tracks on the two-dimensional (2D) plane \cite{Zhu2025}. Each rotation antenna/panel is mounted at an intersection of a horizontal track and a vertical track. With the aid of a motor mounted on the upper side of the array, each vertical track can rotate around the vertical axis, thereby steering all antennas/panels in that column collectively in the azimuth plane. Likewise, motors on the left side of the array allow each horizontal track to rotate about the horizontal axis.  
	To prevent the rotation of the horizontal and vertical tracks from interfering with each other, 
	a dual-axis gimbal structure is employed at each intersection of the tracks. The outer gimbal is rigidly attached to the vertical track and rotates about the vertical axis (azimuth). The inner gimbal, holding the antenna/panel, rotates about the horizontal axis (elevation). The horizontal track drives the elevation axis through a cardan shaft that accommodates the relative motion caused by the azimuth rotation. Similarly, the vertical track directly drives the azimuth axis. With this design, each antenna/panel can be independently steered in both azimuth and elevation by the collective motion of the corresponding vertical and horizontal tracks.

	For conventional RA arrays with antenna-wise independent rotation, the number of motors required for 2D rotation is at least twice the total number of antennas.
	For example, an $M \times N$ 2D array with antenna rotation requires $2MN$ motors with the 2D rotation abilities. In contrast, the required number of motors for the CL-RA architectures is equal to the number of horizontal and vertical tracks, i.e., $M+N$. Nevertheless, even this reduced number may still pose challenges in large-scale deployments. To further lower the hardware cost and control complexity, an antenna panel-level rotation strategy can be adopted, where multiple antenna elements, e.g., $Q_b$, are grouped into a single rotatable panel $b$. Under this configuration, when $\sqrt{Q_b}$ is an integer and $M=N$, the required number of motors is reduced to $(M+N)/\sqrt{Q_b}$.
	Consequently, the considered CL-RA architecture offers a scalable and cost-effective solution, particularly for large-scale arrays where $M$ and $N$ are substantial. 
	
	\begin{figure}[t]
		\centering
		\includegraphics[width=0.48 \textwidth]{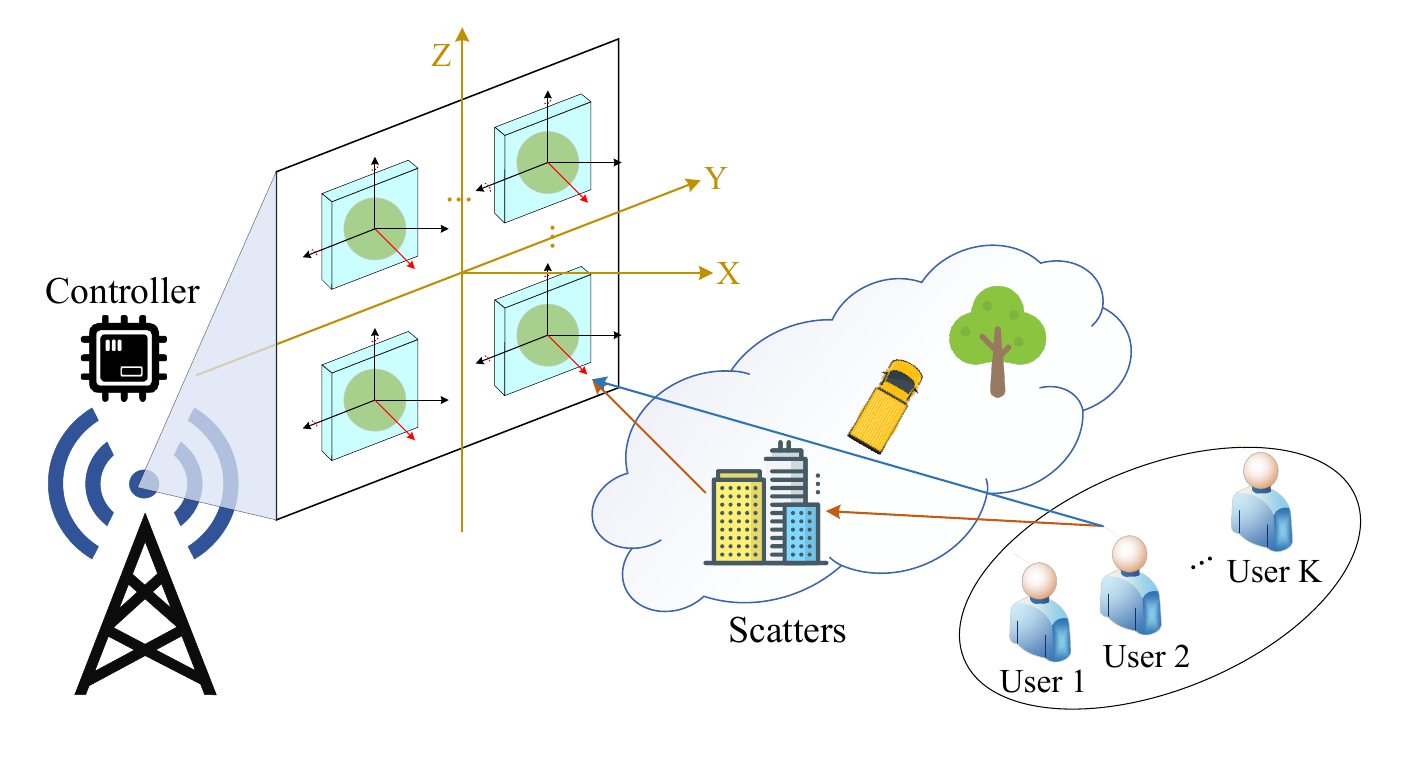}
		\caption{An RA-enabled uplink system.}
		\label{Fig21}
		\vspace{-5mm}
	\end{figure}
	
	\vspace{-3mm}
	\subsection{Antenna Element Rotation-based System}
	\vspace{-1mm}
	We consider an RA-aided uplink system\footnote{While this paper focuses on the uplink system for clarity, by uplink-downlink duality, the proposed algorithmic framework is directly applicable to the downlink communication system.}, where an $M \times N$ RA array is deployed at the BS and the center position of the array is the origin o, as illustrated in Fig. \ref{Fig1} (a). Denote $\mathcal{M}=\{1,2,\ldots, M\}$, $\mathcal{N}=\{1,2,\ldots, N\}$, and $\mathcal{K}=\{1,2,\ldots, K\}$ as the sets of vertical antenna elements, horizontal antenna elements, and users, respectively.
	The position and rotation of the antenna in the $m$-th row  and $n$-th column at the BS, i.e., $a_{m,n}$, can be characterized by $\mathbf{t}_{m,n} = [x_{m,n}, y_{m,n}, z_{m,n}]^\mathrm{T}$ and $\mathbf{u}_{m,n} = [\alpha_{m}, \beta_{n}]$, respectively. The $x_{m,n}$, $y_{m,n}$, and $z_{m,n}$ represent the 3D coordinates of the antenna $a_{m,n}$, satisfying
	\setlength\abovedisplayskip{2pt}
	\setlength\belowdisplayskip{3pt}
	\begin{align}
		&x_{m,n} = 0,y_{m,n} = n\Delta, n = 0, \pm 1, \ldots, \pm \frac{N-1}{2}, \\
		&z_{m,n} = m\Delta, m = 0, \pm 1, \ldots, \pm \frac{M-1}{2}.
	\end{align}
	Here $\alpha_{m} \in (0, 2\pi]$ and $\beta_{n} \in (0, 2\pi]$ denote the rotation angles with respect to the horizontal axis and vertical axis, respectively. 
	Given $\mathbf{u}_{m,n}$, the following rotation matrix can be defined as \cite{Shao2025}
	\begin{align}
		\!\!\!\mathbf{R}(\mathbf{u}_{m,n}) &\!=\! \mathbf{R}_{\alpha_m} \mathbf{R}_{\beta_n} \nonumber \\
		&\!= \!\!\!
		\begin{bmatrix}
			c_{\alpha_m} c_{\beta_n} & c_{\alpha_m} s_{\beta_n} & -s_{\alpha_m} \\
			-s_{\beta_n} & c_{\beta_n} & 0 \\
			s_{\alpha_m} c_{\beta_n} & s_{\alpha_m} s_{\beta_n} & c_{\alpha_m}
		\end{bmatrix}\!, \forall m \in \mathcal{M}, \forall n \in \mathcal{N}, 
	\end{align}
	where $c_x = \cos(x)$ and $s_x = \sin(x)$. $\mathbf{R}_{\alpha_m}$ and $\mathbf{R}_{\beta_n}$ are 
	\begin{align}
		&\!\! \mathbf{R}_{\alpha_m} \!= \!
		\begin{bmatrix}
			c_{\alpha_m} & 0  & -s_{\alpha_m} \\
			0 & 1 & 0 \\
			s_{\alpha_m} & 0 & c_{\alpha_m}
		\end{bmatrix}, 
		\mathbf{R}_{\beta_n}\! = \!
		\begin{bmatrix}
			c_{\beta_n} & s_{\beta_n} & 0 \\
			-s_{\beta_n} & c_{\beta_n} & 0 \\
			0 & 0 & 1
		\end{bmatrix}.
	\end{align}
	
	Each RA element is installed at the cross point of a horizontal track and a vertical track, which is connected to the RF chain via a motor to enable rotation \cite{Zhu2025}. With the aid of $M$ motors mounted on the left side of the array, each row of RAs can collectively rotate around the  horizontal axis, which changes the angle $\alpha_m$ of the antenna $a_{m,n}$. Similarly, with the aid of $N$ motors mounted on the upper side of the array, each vertical track can rotate around the vertical axis, which changes the angle $\beta_n$ of the antenna $a_{m,n}$. 
	Let $\mathbf{u}=[\bm{\alpha},\bm{\beta}]$, where $\bm{\alpha} = [\alpha_1,\alpha_2, \ldots, \alpha_M]$ and $\bm{\beta} = [\beta_1,\beta_2, \ldots, \beta_N]$ denote the vertical and horizontal antenna rotation vectors, respectively. The pointing vector of the antenna $a_{m,n}$ after rotation is given by
	\begin{align}
		&{\mathbf{f}}(\mathbf{u}_{m,n}) = \mathbf{R}(\mathbf{u}_{m,n})\bar{\mathbf{f}}_{m,n}, \forall m \in \mathcal{M}, \forall n \in \mathcal{N},  
	\end{align}
	where $\bar{\mathbf{f}}_{m,n} =[1,0,0]^\mathrm{T} \in \mathbb{R}^{3}$ denotes the initial pointing vector of the antenna $a_{m,n}$.

	\begin{figure}[t]
		\centering
		\includegraphics[width=0.48 \textwidth]{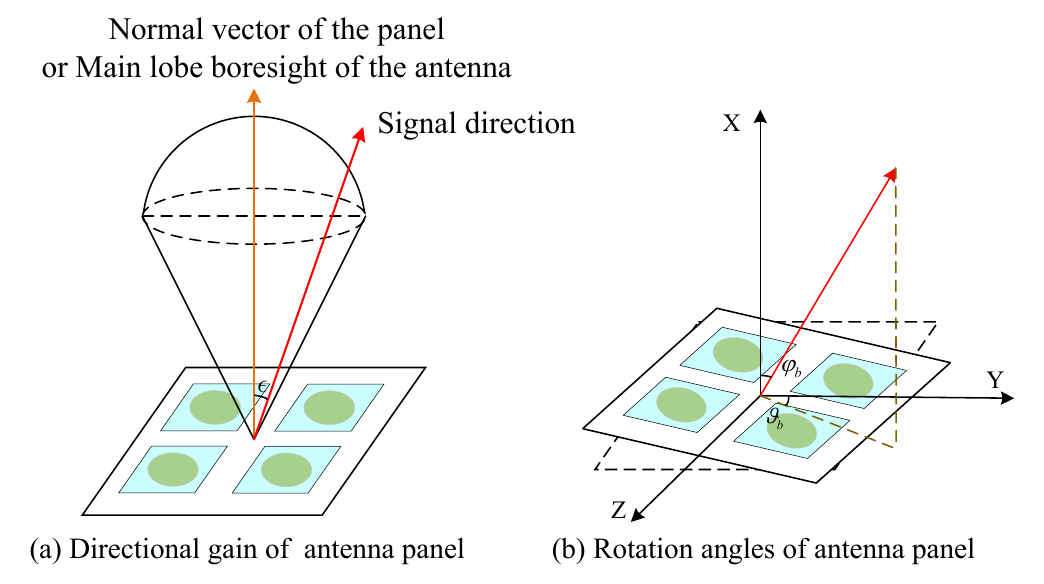}
		\caption{Illustration of directional gain and rotation angles of  antenna panel. If there is only one antenna on the panel, then the normal vector is replaced by the main lobe boresight of the antenna.}
		\label{Fig31}
		\vspace{-5mm}
	\end{figure}
	
	The effective antenna gain for each RA depends on the signal arrival/departure angle and antenna directional gain pattern, as illustrated in Fig. \ref{Fig31}, where $\epsilon$ is the angular offset between the signal direction and the antenna's main-lobe boresight. $\varphi_q$ and $\vartheta_q$ represents the $q$-th antenna's zenith
	angle and azimuth angle, respectively.
	The generic directional gain pattern for each RA is as follows \cite{Peng2025}:
	\begin{align}
		\label{antennaGain}
		G(\epsilon) = 
		\begin{cases}
			G_0\cos^{2p}(\epsilon), & \epsilon \in (0, \frac{\pi}{2}), \\
			0, & \text{otherwise},
		\end{cases}
	\end{align}
	where $\epsilon$ is the elevation angle of any spatial point with respect to the RA's current orientation, $p \ge 0$ determines the directivity of antenna, and $G_0$ is the maximum gain in the boresight direction with $G_0=2(2p+1)$ to satisfy the law of power conservation.
	
	Due to the maneuverability restriction of motors and the need to avoid antenna coupling  between any two RAs, the eccentric angle of the antenna $a_{m,n}$ should be within a specific range, i.e.,
	\begin{align} 
		\label{theta_constarint}
		\cos(\theta_{\max}) \le \cos(\theta_{m,n}) \le 1, \forall m \in \mathcal{M}, \forall n \in \mathcal{N}, 
	\end{align}
	where $\cos(\theta_{m,n}) = {\mathbf{f}}_{m,n}^{\mathrm{T}}\mathbf{e}_x$. Here $\mathbf{e}_x = [1,0,0]^\mathrm{T}$ is the unit vector along the $x$-axis, and $\theta_{\max} \in [0, \pi/2]$.
	
	\subsubsection{Channel Model}
	We consider a narrow-band frequency-flat channel model \cite{Chen2022IRS}.
	Let $\mathbf{v}_k = [x_k,y_k,z_k]^\mathrm{T}$ denote the global coordinates of user $k$. Based on the RA directional gain pattern and Friis transmission equation, the line of sight (LoS) link power gain between antenna $a_{m,n}$ and user $k$ is given by
	\begin{align}
		G_{m,n,k}^{\mathrm{LoS}}(\mathbf{f}_{m,n}) &=\beta_0r_{m,n,k}^{-2}G(\epsilon_{m,n,k}) \nonumber \\ &=\beta_0r_{m,n,k}^{-2}G_0\left(\frac{\mathbf{f}_{m,n}^\mathrm{T}(\mathbf{v}_k-{\mathbf{t}}_{m,n})}{r_{m,n,k}}\right)^{2p},
	\end{align}
	where $\beta_0 = \frac{A}{4\pi }$ is the free-space reference gain constant, $\lambda$ denotes the wavelength, $r_{m,n,k} = \Vert {\mathbf{v}_k - \mathbf{t}}_{m,n}\Vert$ is the distance between antenna $a_{m,n}$ and user $k$, and $\cos(\epsilon_{m,n,k}) =\frac{\mathbf{f}_{m,n}^T(\mathbf{v}_k-{\mathbf{t}}_{m,n})}{r_{m,n,k}}$ represents the cosine of the angle between the boresight $\mathbf{f}_{m,n}$ and the LoS direction to user $k$. The LoS coefficient is the square root of the channel gain with a propagation phase, which is expressed as
	\begin{align}
		h_{m,n,k}^{\mathrm{LoS}}(\mathbf{f}_{m,n})= \sqrt{G_{m,n,k}^{\mathrm{LoS}}(\mathbf{f}_{m,n})}e^{-j\frac{2\pi}{\lambda}r_{m,n,k}}.
	\end{align}
	The channel coefficient captures both the directional antenna gain and the propagation-included phase shift.
	
	We further consider a scattering environment with $D$ spatially distributed clusters, located at $\{\mathbf{s}_d \in \mathbb{R}^3\}^D_{d=1}$. Following the same principle, the non-LoS (NLoS) link power gain between antenna $a_{m,n}$ and cluster $d$ is
	\begin{align}
		G_{m,n,d}^{\mathrm{NLoS}}(\mathbf{f}_{m,n}) &=\beta_0r_{m,n,d}^{-2}G(\epsilon_{m,n,d}) \nonumber \\
		& =\beta_0r_{m,n,d}^{-2}G_0\left(\frac{\mathbf{f}_{m,n}^\mathrm{T}(\mathbf{s}_d-{\mathbf{t}}_{m,n})}{r_{m,n,d}}\right)^{2p},
	\end{align}
	where $r_{m,n,d} =\Vert \mathbf{s}_d-{\mathbf{t}}_{m,n} \Vert$ is the antenna-to-cluster distance, and $\cos(\epsilon_{m,n,d})= \frac{\mathbf{f}_{m,n}^\mathrm{T}(\mathbf{s}_d-{\mathbf{t}}_{m,n})}{r_{m,n,d}}$ denotes the cosine of the angle between the boresight and the direction to cluster $d$. Considering a bi-static scattering model, the NLoS channel coefficient from antenna $a_{m,n}$ to user $k$ is given by
	\begin{align}
		& h_{m,n,k}^{\mathrm{NLoS}}(\mathbf{f}_{m,n}) = \nonumber \\
		&\sum \nolimits_{d=1}^D \sqrt{\frac{G_{m,n,d}^{\mathrm{NLoS}}(\mathbf{f}_{m,n})\sigma_d}{4\pi r^2_{d,k}}} e^{-j\frac{2\pi}{\lambda}(r_{m,n,d}+r_{d,k})+j\chi_d},
	\end{align}
	where $r_{d,k} = \Vert \mathbf{s}_d - \mathbf{u}_k\Vert$ is the cluster-to-user distance, $\sigma_d$ denotes the radar cross section of cluster $d$, and $\chi_d$ is a random phase uniformly distributed over $[0, 2\pi)$. 
	
	Thus, the overall multipath channel between the BS and user $k$ is \cite{Wu2025RA}
	\begin{align}
		\label{channel_h}
		\mathbf{h}_{k} = 
		\mathbf{h}_{k}^{\mathrm{LoS}}(\mathbf{u}) +\mathbf{h}_{k}^{\mathrm{NLoS}}(\mathbf{u}),
	\end{align}
	where 
	\begin{align}
		&\!\! \mathbf{h}_{k}^{\mathrm{LoS}}(\mathbf{u}) = [\hat{\mathbf{h}}_{1,k}^{\mathrm{LoS}}(\mathbf{u}), \hat{\mathbf{h}}_{2,k}^{\mathrm{LoS}}(\mathbf{u}), \ldots, \hat{\mathbf{h}}_{M,k}^{\mathrm{LoS}}(\mathbf{u})]^\mathrm{T},    \\ 
		&\!\! \hat{\mathbf{h}}_{m,k}^{\mathrm{LoS}}(\mathbf{u}) = [h_{m,1,k}^{\mathrm{LoS}}(\mathbf{u}), h_{m,2,k}^{\mathrm{LoS}}(\mathbf{u}), \ldots, h_{m,N,k}^{\mathrm{LoS}}(\mathbf{u})],\\
		&\!\! \mathbf{h}_{k}^{\mathrm{NLoS}}(\mathbf{u}) = [\hat{\mathbf{h}}_{1,k}^{\mathrm{NLoS}}(\mathbf{u}), \hat{\mathbf{h}}_{2,k}^{\mathrm{NLoS}}(\mathbf{u}), \ldots, \hat{\mathbf{h}}_{M,k}^{\mathrm{NLoS}}(\mathbf{u})]^\mathrm{T},   \\
		&\!\! \hat{\mathbf{h}}_{m,k}^{\mathrm{NLoS}}(\mathbf{u}) = [h_{m,1,k}^{\mathrm{NLoS}}(\mathbf{u}), h_{m,2,k}^{\mathrm{NLoS}}(\mathbf{u}), \ldots, h_{m,N,k}^{\mathrm{NLoS}}(\mathbf{u})].
	\end{align}
	
	\vspace{-1mm}
	\subsubsection{Signal Model}
	For the uplink communication, the received signals at the BS is expressed as 
	\begin{align}
		\mathbf{y} =
		\sum\nolimits_{k=1}^K \mathbf{h}_{k}\sqrt{P_k}s_k + \mathbf{n},
	\end{align}
	where $P_k$ and $\mathbf{n} \sim \mathcal{CN}(\mathbf{0}, \sigma^2\mathbf{I}_Q)$ are the transmit power of user $k$ and the additive white Gaussian noise (AWGN), folowing the zero-mean circularly symmetric complex Gaussian (CSCG) distribution with power $\sigma^2$.
	Let $\mathbf{w}_k \in \mathbb{C}^{1 \times Q}$ denote the receive beamforming vector with $\Vert \mathbf{w}_k \Vert =1$. The signal of user $k$ after applying $\mathbf{w}_k$ is given by,
	\begin{align}
		y_k = \mathbf{w}_k\mathbf{h}_{k}\sqrt{P_k}s_k +
		\sum\nolimits_{i=1, i \neq k}^K\mathbf{w}_k \mathbf{h}_{i}\sqrt{P_i}s_i + \mathbf{w}_k\mathbf{n},
	\end{align} 
	where $s_k$ is the information-bearing signal of user $k$.
	Therefore, the SINR of user $k$ is given by
	\begin{align}
		\label{gamma_k}
		\gamma_k = \frac{P_k\vert \mathbf{w}_k\mathbf{h}_{k} \vert^2}{\sum\nolimits_{i=1, i \neq k}^K P_i \vert \mathbf{w}_k\mathbf{h}_{i} \vert^2 + \Vert \mathbf{w}_k\Vert^2\sigma^2}.
	\end{align}
	The data rate of user $k$ is $R_k = \log_2(1+\gamma_k)$.
	
	\vspace{-0mm}
	\subsubsection{Problem Formulation}
	In this paper, we aim to maximize the sum rate of all users by jointly optimizing the receive beamforming matrix $\mathbf{W} = [\mathbf{w}_1,\mathbf{w}_2, \cdots,\mathbf{w}_K]$ and the rotation angle vector $\mathbf{u}$. Thus, the optimization problem is given by
	\begin{subequations}
		\label{P0}
		\begin{eqnarray}
			& \!\!\!\!\!\!\!\!\!\!\!\!\!\!\!\!\! \max  \limits_{\mathbf{w}_k, \mathbf{u}} 
			\label{P0-function}
			& \!\!\!\!\!\!  \sum \limits_{k \in \mathcal{K}} R_k  \\
			\label{P0-C-receive power}  
			& \!\!\!\!\!\!\!\!\!\!\!\!\!\!\!\!\! \mathrm{s.t.} &\!\!\!\!\!\!  \Vert \mathbf{w}_k \Vert  =1, \forall k \in \mathcal{K},\\ 
			\label{P0-C-RA constraint} 
			&&\!\!\!\!\!\! \cos(\theta_{\max}) \le {\mathbf{f}}_{m,n}^{\mathrm{T}}\mathbf{e}_x \le 1, \forall m \in \mathcal{M}, \forall n \in \mathcal{N}, 
		\end{eqnarray}
	\end{subequations}
	where $(\mathrm{\ref{P0-C-RA constraint}})$ denotes the mechanical rotation bounds. Problem $(\mathrm{\ref{P0}})$ is non-convex and challenging to solve because the objective function is quadratic in the beamformers and nonlinear in the orientation variables through the orientation-dependent channels, yielding a nonlinear and non-separable coupling. 
	
	\vspace{-4mm}
	\subsection{Antenna Panel Rotation-based System} \label{Panel Rotation}
	\vspace{-1mm}
	While antenna element-level rotation provides fine control, it incurs prohibitive hardware cost and complexity in large-scale systems due to the multitude of motors required. To address this, we shift to antenna panel-level rotation, where multiple antennas are grouped into a single rotatable panel, as shown in Fig. \ref{Fig1} (b). This strategy significantly reduces the number of motors and control overhead, while still leveraging the panel's inherent array gain and beamforming capabilities.
	Specifically, we assume that the $Q$-antenna BS is divided into $B$ rotatable panels, with each panel comprises $Q_b$ elements, i.e., $ Q= B \times Q_b$ and $B= M \times N$. 
	The set of panels is given by $\mathcal{B}=\{1,2,\ldots, B\}$. Besides, th number of antenna at each panel satisfies $Q_b = M_b \times N_b$. The position and rotation of the $b$-th panel at the BS can be characterized by $\mathbf{q}_b = [x_b, y_b, z_b]^T$ and $\mathbf{u}_{m,n}^b = [\alpha_{m}^b, \beta_{n}^b]$, respectively, where $x_b$, $y_b$, and $z_b$ represent the coordinates of the $b$-th panel center in the global Cartesian coordinate system (CCS).
	Here $\alpha_{m}^b \in (0, 2\pi]$ and $\beta_{n}^b \in (0, 2\pi]$ denote the rotation angles of panel $b$ with respect to the $y$-axis and $z$-axis, respectively. Then, we have the rotation matrix as $\mathbf{R}(\mathbf{u}_{m,n}^b) = \mathbf{R}_{\alpha_{m}^b} \mathbf{R}_{\beta_{n}^b}$.

	Each panel is installed at the cross point of a horizontal track and a vertical track. By tuning the motors, the angles $\alpha_{m}^b$ and $\beta_{n}^b$ can be changed. 
	In addition, each panel's local CCS is denoted by $o'-x'y'z'$, with the surface center serving as the origin $o'$. The $x'$-axis is oriented along the direction of the normal vector of the panel. Let $\bar{\mathbf{r}}_{q}$ denote the position of the $q$-th antenna of the panel in its local CCS. Then, the position of the $q$-th antenna of the $b$-th panel in the global CCS, denoted by ${\mathbf{r}}_{b,q}$, is given by
	\begin{align}
		{\mathbf{r}}_{b,q} = \mathbf{q}_{b} + \mathbf{R}(\mathbf{u}_{m,n}^b)\bar{\mathbf{r}}_{q}.
	\end{align}
	
	\vspace{-2mm}
	\subsubsection{Channel Model}
	Denote the outward normal vector of the $b$-th panel as
	\begin{align}
		\mathbf{n}(\mathbf{u}_{m,n}^b) =\mathbf{R} (\mathbf{u}_{m,n}^b)\bar{\mathbf{n}}, \forall m \in \mathcal{M}, \forall n \in \mathcal{N}, 
	\end{align}
	where $\bar{\mathbf{n}} = [1,0,0]^\mathrm{T}$ represents the normal vector of the $b$-th panel in the local CCS.

	Let $a_{b,q}$ represent the $q$-th antenna of the $b$-th panel, the link power gain between antenna $a_{b,q}$ and user $k$ is given by
	\begin{align}
		\!\!\! \widetilde{G}_{b,q,k}^{\mathrm{LoS}}(\mathbf{n}(\mathbf{u}_{m,n}^b)) &\!\!=\!\!\beta_0r_{b,q,k}^{-2}G(\epsilon_{b,q,k}) \nonumber \\ &\!\!=\!\!\beta_0r_{b,q,k}^{-2}G_0\!\! \left(\!\!\frac{\mathbf{n}(\mathbf{u}_{m,n}^b)^{\mathrm{T}}(\mathbf{v}_k \!-\! {\mathbf{r}}_{b,q})}{r_{b,q,k}}\!\!\right)^{2p},
	\end{align}
	where $r_{b,q,k} = \Vert {\mathbf{v}_k - \mathbf{r}}_{b,q}\Vert$ is the distance between antenna $a_{b,q}$ and user $k$. As illustrated in Fig. \ref{Fig31}, $\cos(\epsilon_{b,q,k}) =\frac{\mathbf{n}(\mathbf{u}_{m,n}^b)^{\mathrm{T}}(\mathbf{v}_k-{\mathbf{r}}_{b,q})}{r_{b,q,k}}$ represents the cosine of the angle between the normal vector $\mathbf{n}(\mathbf{u}_{m,n}^b)$ and the LoS direction to user $k$. Then, the LoS coefficient is expressed as
	\begin{align}
		\widetilde{h}_{b,q,k}^{\mathrm{LoS}}(\mathbf{n}(\mathbf{u}_{m,n}^b))= \sqrt{\widetilde{G}_{b,q,k}^{\mathrm{LoS}}(\mathbf{n}(\mathbf{u}_{m,n}^b))}e^{-j\frac{2\pi}{\lambda}r_{b,q,k}}.
	\end{align}
	For NLoS path, the link power gain between antenna $a_{b,q}$ and cluster $d$ is
	\begin{align}
		\!\!\!\widetilde{G}_{m,n,d}^{\mathrm{NLoS}}(\mathbf{n}(\mathbf{u}_{m,n}^b)) &\!\!=\!\!\beta_0r_{b,q,d}^{-2}G(\epsilon_{m,n,d}) \nonumber \\
		& \!\!=\!\!\beta_0r_{b,q,d}^{-2}G_0\!\!\left(\!\!\frac{\mathbf{n}(\mathbf{u}_{m,n}^b)^{\mathrm{T}}(\mathbf{s}_d \!-\!{\mathbf{r}}_{b,q})}{r_{b,q,d}}\!\!\right)^{2p},
	\end{align}
	where $r_{b,q,d} =\Vert \mathbf{s}_d-{\mathbf{r}}_{b,q} \Vert$ is the antenna-to-cluster distance, and $\cos(\epsilon_{b,q,d})= \frac{\mathbf{n}(\mathbf{u}_{m,n}^b)^{\mathrm{T}}(\mathbf{s}_d-{\mathbf{r}}_{b,q})}{r_{b,q,d}}$ denotes the cosine of the angle between the antenna boresight and the direction to cluster $d$. The NLoS channel coefficient from antenna $a_{b,q}$ to user $k$ is given by
	\begin{align}
		&\widetilde{h}_{b,q,k}^{\mathrm{NLoS}}(\mathbf{n}(\mathbf{u}_{m,n}^b)) = \nonumber \\  
		&\sum \nolimits_{d=1}^D  \sqrt{\frac{\widetilde{G}_{b,q,d}^{\mathrm{NLoS}}(\mathbf{n}(\mathbf{u}_{m,n}^b))\sigma_d}{4\pi r^2_{d,k}}} e^{-j\frac{2\pi}{\lambda}(r_{b,q,d}+r_{d,k})+j\chi_d}.
	\end{align}
	Based on the above, similar to $(\ref{channel_h})$, the channel from the user $k$ to the BS is given by
	\begin{align}
		\widetilde{\mathbf{h}}_k = 
		\widetilde{\mathbf{h}}_k^{\mathrm{LoS}}(\mathbf{u}) +\widetilde{\mathbf{h}}_k^{\mathrm{NLoS}}( \mathbf{u}).
	\end{align}
	\subsubsection{Signal Model}
	The received signals at the BS are given by 
	\begin{align}
		\widetilde{\mathbf{y}} =
		\sum\nolimits_{k=1}^K \widetilde{\mathbf{h}}_k\sqrt{P_k}s_k + \mathbf{n}.
	\end{align}
	Therefore, the data rate of user $k$ is $\widetilde{R}_k = \log_2(1+\widetilde{\gamma}_k)$, where
	\begin{align}
		\label{gamma_k1}
		\widetilde{\gamma}_k = \frac{P_k\vert \mathbf{w}_k\widetilde{\mathbf{h}}_{k} \vert^2}{\sum\nolimits_{i=1, i \neq k}^K P_i \vert \mathbf{w}_k\widetilde{\mathbf{h}}_{i} \vert^2 + \Vert \mathbf{w}_k\Vert^2\sigma^2}.
	\end{align}
	
	In order to avoid mutual signal reflection between different panels, the panel $b$ should satisfy
	\begin{align}
		\mathbf{n}(\mathbf{u}_{m,n}^b)^{\mathrm{T}}(\mathbf{q}_j-\mathbf{q}_b) \le 0,
	\end{align}
	where $j \neq b$ and $j \in \mathcal{B}$.
	In addition, to avoid the $b$-th panel to rotate towards the central processing unit (CPU) of the BS, it should ensure
	\begin{align}
		\mathbf{n}(\mathbf{u}_{m,n}^b)^{\mathrm{T}}\mathbf{q}_b \ge 0.
	\end{align}
	
	\subsubsection{Problem Formulation}
	Based on the above signal model, the sum rate maximization problem is formulated as follows:
	\begin{subequations}
		\label{P1-0}
		\begin{eqnarray}
			& \!\!\!\!\!\!\!\!\!\!\!\!\!\!\!\!\! \max  \limits_{\mathbf{w}_k, \mathbf{u}} 
			\label{P1-0-function}
			& \!\!\!\!\!\!  \sum \limits_{k \in \mathcal{K}} \widetilde{R}_k  \\
			\label{P1-0-C-receive power}  
			& \!\!\!\!\!\!\!\!\!\!\!\!\!\!\!\!\! \mathrm{s.t.} &\!\!\!\!\!\!  \Vert \mathbf{w}_k \Vert  =1, \forall k \in \mathcal{K},\\ 
			\label{P1-0-C-RA constraint} 
			&&\!\!\!\!\!\! 	\mathbf{n}(\mathbf{u}_{m,n}^b)^{\mathrm{T}}(\mathbf{q}_j-\mathbf{q}_b) \le 0, \forall b \in \mathcal{B}, \forall j \in \mathcal{B}, j \neq b,\\
			\label{P1-0-C-RA signal leakage} 
			&&\!\!\!\!\!\!
			\mathbf{n}(\mathbf{u}_{m,n}^b)^{\mathrm{T}}\mathbf{q}_b \ge 0, \forall b \in \mathcal{B},
		\end{eqnarray}
	\end{subequations}
	where constraint $(\mathrm{\ref{P1-0-C-RA constraint}})$ avoids mutual signal reflection between any two panels. Constraint $(\mathrm{\ref{P1-0-C-RA signal leakage}})$ prevents each panel from rotating towards the CPU which causes signal blockage. Problem $(\mathrm{\ref{P1-0}})$ is challenging to solve because the objective function is non-concave, as well as constraints $(\mathrm{\ref{P1-0-C-RA constraint}})$ and $(\mathrm{\ref{P1-0-C-RA signal leakage}})$ are non-linear. 
	
	\subsubsection{Feasible Rotation Angle Range Analysis}
	Note that the position of each panel at the BS is fixed. Given the practical physical constraints $(\mathrm{\ref{P1-0-C-RA constraint}})$ and $(\mathrm{\ref{P1-0-C-RA signal leakage}})$, the feasible rotation angles set of panels is much smaller than that of a single panel. Specifically, for a single panel, its vertical and horizontal rotation angles can be flexibly adjusted to obtain the best performance. However, when an array is divided into several panels, the panels need to satisfy the practical physical constraints among themselves. 
	\begin{proposition}
		\label{pro2}
		\rm{To meet the physical constraints, the panel $b$ should ensure:
			\begin{align}
				\label{Constr1}
				& -s_{\beta_n}(n_j-n)+s_{\alpha_m} c_{\beta_n}(m_j-m) \le 0, \\
				\label{Con_self}
				&  s_{\alpha_m} c_{\beta_n}m \ge s_{\beta_n}n,
			\end{align}
			where $\forall j \neq b, j \in \mathcal{B}$ and $j = (m_j-1)N+n_j$ and $b=(m-1)N+n$.}
	\end{proposition}
	\begin{proof}
		Please see Appendix A.
	\end{proof}
	\begin{remark}
		\rm{We first consider a $1 \times N$ configuration. For the panel located on the left side (i.e., $n=-\frac{N-1}{2}$), we have $m_j = 0$ and $n_j > n $. Based on $(\ref{Constr1})$ and $(\ref{Con_self})$, we get $s_{\beta_n} \ge 0$. Thus, $\beta_n \in [0,\pi]$ and $\alpha_m  \in [0,2\pi)$. For the panel located on the middle corner (i.e., $-\frac{N-1}{2}< n < \frac{N-1}{2}$), we have $s_{\beta_n} = 0$. Thus, we obtain $\beta_n = 0$ and $\alpha_m  \in [0,2\pi)$. Then, we consider the $M \times 1$ configuration. For the panel located on the bottom side (i.e., $m=-\frac{M-1}{2}$ ), we have $n_j =0$ and $m_j > m $. Then, we get $s_{\alpha_m}c_{\beta_n} \le 0$. When $c_{\beta_n} = 0$, the outward norm vector is fixed with $\mathbf{n} = [0,\pm 1,0]^{\mathrm{T}}$ for all $\alpha_m$ value in this case. If $c_{\beta_n} > 0$ (i.e., $\beta_n \in [0,\pi/2) \cup (3\pi/2,2\pi]$), $s_{\alpha_m} \le 0$ (i.e., $\alpha_m \in [\pi,2\pi]$). Otherwise, $\beta_n \in (\pi/2,3\pi/2)$ and $\alpha_m \in [0,\pi]$.
			For the panel located on the middle corner (i.e., $-\frac{M-1}{2}< m < \frac{M-1}{2}$), we have $s_{\alpha_m}c_{\beta_n} = 0$. If $c_{\beta_n} = 0$, $\alpha_m \in [0,2\pi)$. However, the outward norm vector is fixed with $\mathbf{n} = [0,\pm 1,0]^{\mathrm{T}}$ for all $\alpha_m$ value in this case. When $s_{\alpha_m} = 0$ (i.e., $\alpha_m = 0$), $\beta_n \in [0,2\pi)$.

			For $M \times N$ configuration, where $M = N > 1$, if a panel is surrounded by adjacent panels on all four sides (e.g., top, bottom, left, and right), its normal vector must align exclusively with the $+x$ direction to meet the physical constraints. 
			For panels located at the four corners, the rotation angles must satisfy specific conditions. For example, if panel $b$ is located in the bottom right corner (i.e., $m=-\frac{M-1}{2}$ and $n=\frac{N-1}{2}$), we have $m_j > m $ and $n_j < n, \forall j \neq b$. Then, the overall rotation angle constraint is given by
			\begin{align}
				&  \!\!\! s_{\alpha_m} c_{\beta_n} \!\le\! -s_{\beta_n}, 
				s_{\beta_n}(n-n_j)\!+\! s_{\alpha_m} c_{\beta_n}(m_j-m) \! \le \! 0.
			\end{align}
			Since there exists $m_j = m$ and $n_j = n$, we can obtain a subset of the above constraints, which is 
			\begin{align}
				s_{\beta_n} \le 0, s_{\alpha_m} c_{\beta_n} \le 0.
			\end{align}
			Then, if $c_{\beta_n} = 0$, the outward norm vector is fixed with $\mathbf{n} = [0,\pm 1,0]^{\mathrm{T}}$ for all $\alpha_m$ value in this case. When $c_{\beta_n} > 0$ (i.e., $\beta_n \in [3\pi/2,2\pi)$ ), we have $s_{\alpha_m} \le 0$, where $\alpha_m \in [\pi,2\pi]$. When $c_{\beta_n} < 0$ (i.e., $\beta_n \in [\pi,3\pi/2)$ ), we have $s_{\alpha_m} \ge 0$, where $\alpha_m \in [0,\pi]$. }
	\end{remark}
	
	According to the above discussions, it can be observed that when the positions of multiple panels are fixed, the feasible rotation angles range of these panels is more restricted than that of a single panel.
	
	\vspace{-4mm}
	\section{Single-User Case} \label{Performance Analysis}
	\vspace{-2mm}
	In this section, we consider the single-user case at the antenna element-level rotation, with $K=1$ and $D=0$. The sum rate maximization problem (\ref{P0}) can be reformulated as 
	\begin{subequations}
		\label{P3-0}
		\begin{eqnarray}
			&\max  \limits_{\mathbf{w}, \mathbf{u}} 
			\label{P3-0-function}
			&  \bar{P}\vert \mathbf{w}\mathbf{h}^{\mathrm{LoS}} \vert^2  \\
			\label{P3-0-C-receive power}  
			&  \mathrm{s.t.} &   \Vert \mathbf{w} \Vert  =1,(\mathrm{\ref{P0-C-RA constraint}}),
		\end{eqnarray}
	\end{subequations}
	where $\bar{P} = P /\sigma^2$.
	When the rotation angle vector $\mathbf{u}$ is fixed, we apply the maximum ratio combining (MRC) beamformer to obtain the optimal receive beamforming vector $\mathbf{w}^{\ast} = \frac{\mathbf{h}^{\mathrm{LoS}}}{\Vert \mathbf{h}^{\mathrm{LoS}} \Vert}$. Then, we have $\bar{P}\vert \mathbf{w}\mathbf{h}^{\mathrm{LoS}} \vert^2 = \bar{P}\Vert \mathbf{h}^{\mathrm{LoS}} \Vert^2 $. By introducing $\mathbf{w}^{\ast}$ into $(\ref{P3-0-function})$, the SNR of the user is given by
	\begin{align}
		\label{single_gamma}
		\gamma = \bar{P} \beta_0 G_0\sum \limits_{m \in M}\sum \limits_{n \in N} \frac{\cos(\epsilon_{m,n})^{2p}}{r_{m,n}^2}.
	\end{align}
	Based on the above, problem $(\ref{P3-0})$ is reformulated as
	\begin{subequations}
		\label{P3-1}
		\begin{eqnarray}
			& \max  \limits_{\mathbf{u}} 
			\label{P3-1-function}
			& \sum \limits_{m \in M}\sum \limits_{n \in N} \frac{\cos(\epsilon_{m,n})^{2p}}{r_{m,n}^2} \\
			\label{P3-1-C-receive power}  
			& \mathrm{s.t.} & (\mathrm{\ref{P0-C-RA constraint}}).
		\end{eqnarray}
	\end{subequations}
	Based on the CL-RA architecture, the antennas on each row have the same rotation angle $\alpha$, and the antennas on each column have the same rotation angle $\beta$. Then, we have the pointing vector of the antenna $a_{m,n}$ as ${\mathbf{f}}(\mathbf{u}_{m,n}) = \mathbf{R}(\mathbf{u}_{m,n})\bar{\mathbf{f}}_{m,n} = [c_{\alpha_m}c_{\beta_n} \ -s_{\beta_n} \ s_{\alpha_m} c_{\beta_n} ]^\mathrm{T}$. In order to make each antenna beam point towards the user, where the 3D position of the user is $\mathbf{v}_0 = [x_0, y_0, z_0]$, the unit vector pointing from antenna $a_{m,n}$ towards the user is
	\begin{align}
		\mathbf{u}_0& =\frac{(x_0-x_{m,n},y_0-y_{m,n},z_0-z_{m,n})}{\sqrt{(x_0-x_{m,n})^2+(y_0-y_{m,n})^2+(z_0-z_{m,n})^2}} \nonumber \\
		&=\frac{(x_0,y_0-y_n,z_0-z_m)}{\sqrt{x_0^2+(y_0-y_n)^2+(z_0-z_m)^2}},
	\end{align}
	where $y_n=n\Delta$ and $z_m=m\Delta$.

	\begin{proposition}
		\label{pro1}
		\rm{For the ULA array, i.e., $1 \times N$, by adjusting the rotation angles, each antenna in the CL-RA architecture can point towards the user, where the rotation angles are given by
			\begin{align}
				\label{cos_alpha1}
				& \alpha_m = \mathrm{\arccos}\left(\frac{x_0}{\sqrt{x_0^2+(z_0-z_m)^2}}\right),\\
				\label{cos_beta1}
				& \beta_n = \mathrm{\arccos}\left(\frac{\sqrt{x_0^2+(z_0-z_m)^2}}{r}\right).
		\end{align}}
	\end{proposition}
	\begin{proof}
		Please refer to Appendix B.
	\end{proof}
	
	Proposition \ref{pro1} reveals that under the ULA setup, the CL-RA architecture is able to attain equivalent performance to its fully flexible counterpart, while reducing the number of required motors from $2N$ to $N+1$. 
	The same conclusion holds for a $M \times 1$ array, where the rotation sequence is changed (i.e., rotating around the $y$-axis first, followed by the $z$-axis). However, for the UPA array, the limited rotational DoFs restricts the performance of the CL-RA architecture compared to the fully flexible counterpart.
	
	We consider constraint $(\mathrm{\ref{P0-C-RA constraint}})$ under the ULA case, i.e. $1 \times N$, where $m=1$. Since
	${\mathbf{f}}(\mathbf{u}_{m,n}) = \mathbf{R}(\mathbf{u}_{m,n})\bar{\mathbf{f}}_{m,n}$, constraint $(\mathrm{\ref{P0-C-RA constraint}})$ can be converted as $\cos(\theta_{\max}) \le \cos(\alpha_m)\cos(\beta_n) \le 1$. Considering  $\cos(\alpha_m)\cos(\beta_n) \le 1$, the inequality on the right hand is guaranteed. For $\cos(\theta_{\max}) \le \cos(\alpha_m)\cos(\beta_n)$, we have 
	\begin{align}
		\cos(\theta_{\max}) \le \cos(\alpha_m)\cos(\beta_n) = \frac{x_0}{r}.
	\end{align}
	Then, for the antenna $a_{m,n}$, the eccentric angle to direct to the user is expressed as 
	\begin{align}
		\theta_{m,n} = \arctan\left(\frac{\sqrt{(y_0-y_n)^2+(z_0-z_m)^2}}{x_0}\right).
	\end{align}
	When $\theta_{m,n} \le \theta_{\max}$, we have $\epsilon_{m,n} = 0$ and the antenna $a_{m,n}$ can be aligned with the user to achieve the best performance \cite{Zheng2025-2}. When $\theta_{m,n} \ge \theta_{\max}$, $\epsilon_{m,n} = \theta_{m,n} - \theta_{\max}$. Thus, we have
	\begin{align}
		\label{singel-epsilon}
		\cos(\epsilon_{m,n}) = \cos([\theta_{m,n} - \theta_{\max}]^{+}),
	\end{align}
	where $[x]^{+}=\max\{0,x\}$. 
	
	\vspace{-2mm}
	\section{Multi-User Case} \label{Proposed Algorithm}
	In this section, we consider the more general multi-user scenario. Due to the non-convex nature of the problem and coupled constraints, it is challenging to optimize antenna rotation angles and receive beamforming. Therefore, we propose an AO algorithm to address the formulated problem in the following.
	\subsection{Proposed Algorithm for Antenna Optimization}
	In this subsection, we apply an AO algorithm to solve problem $(\ref{P0})$.
	\subsubsection{Receive Beamforming Optimization}  \label{Receive beamforming Optimization}
	With fixed $\mathbf{u}$, we first optimize the receive beamforming $\mathbf{w}_k$. The related sub-problem is given by
	\begin{subequations}
		\label{P1-1}
		\begin{eqnarray}
			&  \max  \limits_{\mathbf{w}_k} 
			\label{P1-1-function}
			&  \sum \limits_{k \in \mathcal{K}} R_k  \\
			\label{P1-1-C-receive power}  
			&  \mathrm{s.t.} &  \Vert \mathbf{w}_k \Vert  =1, \forall k \in \mathcal{K}.
		\end{eqnarray}
	\end{subequations}
	Then, we have 
	\begin{align}
		\gamma_k = \frac{\bar{P}_k\vert \mathbf{w}_k\mathbf{h}_{k} \vert^2}{\sum\nolimits_{i=1, i \neq k}^K \bar{P}_i \vert \mathbf{w}_k\mathbf{h}_{i} \vert^2 + 1}.
	\end{align}
	By exploiting the MMSE beamforming, we can obtain the optimal solution of problem $(\ref{P1-1})$ as
	\begin{align}
		\label{W-MMSE}
		\mathbf{w}_k^{\mathrm{MMSE}} = \frac{\mathbf{C}_k^{-1}\mathbf{h}_k}{\Vert \mathbf{C}_k^{-1}\mathbf{h}_k \Vert},
	\end{align}
	where $\mathbf{C}_k = \sum\nolimits_{i=1, i \neq k}^K \bar{P}_i  \mathbf{h}_{i}\mathbf{h}_{i}^{\mathrm{H}} + \mathbf{I}_{N}$ is the interference-plus-noise-covariance matrix. To reduce the dimension of matrix inversion from $Q \times Q$ to $(K-1) \times (K-1)$, $\mathbf{C}_k^{-1}$ can be converted by applying the Woodbury matrix identity:
	\begin{align}
		\mathbf{C}_k^{-1} &= (\mathbf{I}_{N} + \widetilde{\mathbf{H}}_{k}\mathbf{P}_{k}\widetilde{\mathbf{H}}_{k}^{\mathrm{H}})^{-1} \nonumber \\
		& = \mathbf{I}_{N} -\widetilde{\mathbf{H}}_{k} (\mathbf{P}_{k}^{-1}+\widetilde{\mathbf{H}}_{k}^{\mathrm{H}}\widetilde{\mathbf{H}}_{k})^{-1}\widetilde{\mathbf{H}}_{k}^{\mathrm{H}},
	\end{align}
	where $\widetilde{\mathbf{H}}_{k} = [\mathbf{h}_1,\ldots,\mathbf{h}_{k-1},\mathbf{h}_{k+1},\ldots,\mathbf{h}_{K}]$ and $\mathbf{P}_k = \mathrm{diag}(\bar{P}_1,\ldots,\bar{P}_{k-1},\bar{P}_{k+1},\ldots,\bar{P}_K)$. 

	\subsubsection{Rotation Angle Optimization}
	With fixed $\mathbf{w}_k$, the rotation angles optimization subproblem is given by
	\begin{subequations}
		\label{P2}
		\begin{eqnarray}
			&\!\!\!\!\!\!\!\!\! \max  \limits_{\mathbf{u}} 
			\label{P2-function}
			&\!\!\! \sum \limits_{k \in \mathcal{K}} R_k  \\
			\label{P2-C-RA constraint} 
			&\!\!\!\!\!\!\!\!\! \mathrm{s.t.} &\!\!\!  \cos(\theta_{\max}) \le {\mathbf{f}}_{m,n}^{\mathrm{T}}\mathbf{e}_x \le 1, \forall m \in \mathcal{M}, \forall n \in \mathcal{N}, 
		\end{eqnarray}
	\end{subequations}
	where the objective function $(\mathrm{\ref{P2-function}})$ is non-concave, and constraint $(\mathrm{\ref{P2-C-RA constraint}})$ is non-convex, which makes it challenging to solve problem $(\mathrm{\ref{P2}})$ optimally. We apply the feasible direction method to solve problem $(\mathrm{\ref{P2}})$.
	Denote $\mathbf{u}_{m,n}^{(t-1)} = [\alpha_m^{(t-1)}, \beta_n^{(t-1)}]$ as the rotation solution after iteration $t-1$. Let $\Delta\mathbf{u}_{m,n} = \mathbf{u}_{m,n} - \mathbf{u}_{m,n}^{(t-1)} = [\Delta\alpha_m,\Delta\beta_n], \forall m,n$,
	as the corresponding increments in the $t$-th iteration, where $\mathbf{u}_{m,n}^{(t-1)}$ is the value of $\mathbf{u}_{m,n}$ in the $t-1$ iteration. The terms $\Delta\alpha_m =\alpha_m - \alpha_m^{(t-1)} $ and $\Delta\beta_n =\beta_n - \beta_n^{(t-1)} $. Then, the update for the rotation matrix at the current iteration can be expressed as the product of the rotation matrix from the previous iteration, denoted as $\mathbf{R}(\mathbf{u}_{m,n}^{(t-1)})$, and the incremental rotation matrix, denoted as $\mathbf{R}(\Delta\mathbf{u}_{m,n})$, which is given by
	\begin{align}
		\label{R1}
		\mathbf{R}(\mathbf{u}_{m,n}) =\mathbf{R}(\mathbf{u}_{m,n}^{(t-1)}) \mathbf{R}(\Delta\mathbf{u}_{m,n}),\forall m \in \mathcal{M}, \forall n \in \mathcal{N}, 
	\end{align}
	where $\mathbf{R}(\mathbf{u}_{m,n})$ is a unitary matrix belonging to the orthogonal group.
	Considering that the angle changes $\Delta\mathbf{u}_{m,n}$ are very small in each iteration, we can apply the following small-angle approximations:
	\begin{align}
		\label{cos}
		\cos(x) \rightarrow 1,  
		\sin(x) \rightarrow x,
	\end{align}
	for $x \rightarrow 0$. Substituting $\mathbf{u}_{m,n}$ with $\Delta\mathbf{u}_{m,n}$ and applying the linearization approximations $(\mathrm{\ref{cos}})$, we can obtain the following linear approximation:
	\begin{align}
		\label{R_new}
		\!\!\! \mathbf{R}(\Delta\mathbf{u}_{m,n})  \!\approx\! \mathbf{R}_{\Delta\alpha_m} \mathbf{R}_{\Delta\beta_n}  
		\!\!= \!\!
		\begin{bmatrix}
			1 & \Delta\beta_n & -\Delta\alpha_m \\
			-\Delta\beta_n & 1 & 0 \\
			\Delta\alpha_m & 0 & 1
		\end{bmatrix}.
	\end{align}
	Substituting $(\mathrm{\ref{R1}})$ and $(\mathrm{\ref{R_new}})$ into ${\mathbf{f}}_{m,n}$, we can linearize the non-convex constraint $(\mathrm{\ref{P2-C-RA constraint}})$ as follows:
	\begin{align}
		\label{Constraint1}
		& \!\!\!\!\!\!\!\!\! \cos(\theta_{\max}) \le \bar{\mathbf{f}}_{m,n}^{\mathrm{T}}\mathbf{R}(\Delta\mathbf{u}_{m,n})^{\mathrm{T}}\mathbf{R}(\mathbf{u}_{m,n}^{(t-1)})^{\mathrm{T}}\mathbf{e}_x \le 1.
	\end{align}
	Then, we exploit the feasible direction methods, which starts with a feasible vector $\mathbf{u}^{(0)}$ and generates a sequence of feasible vectors $\mathbf{u}^{(t-1)}$. We get
	\begin{align}
		\mathbf{u}^{(t)} = \mathbf{u}^{(t-1)} + \iota^{(t-1)}(\bar{\mathbf{u}}^{(t-1)} - \mathbf{u}^{(t-1)}),
	\end{align}
	where $\iota^{(t-1)} \in (0,1]$ is the adaptive step size calculated by the Armio rule. Here, $\bar{\mathbf{u}}^{(t-1)}$ is a feasible vector, which can be chosen as the solution to the following optimization problem:
	\vspace{-4mm}
	\begin{subequations}
		\label{P3}
		\begin{eqnarray}
			& \max  \limits_{\mathbf{u}} 
			\label{P3-function}
			& \nabla_{\mathbf{u}} C_1(\mathbf{u}^{(t-1)})^{\mathrm{T}} (\mathbf{u} - \mathbf{u}^{(t-1)}) \\
			\label{P3-C-RA constraint} 
			& \mathrm{s.t.} & (\mathrm{\ref{Constraint1}}),
		\end{eqnarray}
	\end{subequations}
	where $C_1 = \sum \nolimits_{k \in \mathcal{K}} R_k$. The gradient of function $C_1(\mathbf{u}^{(t-1)})$ at the point $\mathbf{u}^{(t-1)}$ is given by
	\begin{align}
		\label{u-gradient}
		&\!\!\! [\nabla_{\mathbf{u}} C_1(\mathbf{u}^{(t-1)})]_j \nonumber \\
		&\!\!\! = \!\lim_{\varepsilon \rightarrow 0 }\frac{C_1(\mathbf{u}^{(t-1)} \!+\! \varepsilon\mathbf{e}^j)\!-\!C_1(\mathbf{u}^{(t-1)})}{\varepsilon},1 \le j \le M+N.
	\end{align}
	where $\mathbf{e}^j \in \mathbb{R}^{M+N}$ is a vector where the $j$-th element is 1, and all other elements are $0$. Note that problem $(\mathrm{\ref{P3}})$ is a linear program and can be optimally solved by using linprog \cite{Rocha2019}. 
	The detailed algorithm for solving problem $(\ref{P3})$ is presented in Algorithm \ref{Algorithm1}.
	
	\begin{algorithm}[t]  
		\caption{Feasible direction algorithm for Solving (\ref{P2})}
		\label{Algorithm1}
		\renewcommand{\algorithmicrequire}{\textbf{Initialize}}
		\renewcommand{\algorithmicensure}{\textbf{Output}}
		\begin{algorithmic}[1]
			\STATE \textbf{Initialize} $\mathbf{u}^{(0)}$, $\rho$, $\upsilon$, step size $\iota_0$, the maximum number of iteration $T_{\max}$, and the predefined threshold $\epsilon_0$.
			\REPEAT
			\STATE Set the iteration index $\tau=1$;
			\STATE Compute the gradient value $\nabla_{\mathbf{u}} C_1(\mathbf{u}^{(\tau-1)})$ based on (\ref{u-gradient}) and set $\iota = \iota_0$;
			\STATE Obtain $\bar{\mathbf{u}}^{(\tau-1)}$ by solving problem (\ref{P3});
			\REPEAT 
			\STATE Compute $\mathbf{u}^{(\tau)} = \mathbf{u}^{(\tau-1)} + \iota(\bar{\mathbf{u}}^{(\tau-1)} - \mathbf{u}^{(\tau-1)})$
			\STATE Update $\iota = \rho \iota$;
			\UNTIL $ C_1(\mathbf{u}^{(\tau)}) -  C_1(\mathbf{u}^{(\tau-1)}) \ge \upsilon\iota \nabla_{\mathbf{u}}  C_1(\mathbf{u}^{(\tau-1)})^{\mathrm{T}}(\bar{\mathbf{u}}^{(\tau-1)} - \mathbf{u}^{(\tau-1)}) $ and $\mathbf{u}^{(\tau)}$ is feasible;
			\STATE Update $\tau=\tau+1$;
			\UNTIL $\nabla_{\mathbf{u}} C_1(\mathbf{u}^{(\tau)})^\mathrm{T}({\mathbf{u}}^{(\tau)} - \bar{\mathbf{u}}^{(\tau)}) \le \epsilon_0$ or $\tau \ge T_{\max}$;
			\STATE \textbf{Output} the converged $\mathbf{u}$.
		\end{algorithmic}
	\end{algorithm}	
	
	\vspace{-3mm}
	\subsection{Proposed Algorithm for Panel Optimization}
	In this subsection, we propose an algorithm similar to that for antenna rotation optimization to address problem $(\ref{P1-0})$.
	The receive beamforming can be obtained by the proposed MMSE method in Section \ref{Receive beamforming Optimization}.
	Then, with fixed $\mathbf{w}_k$, we optimize $\mathbf{u}$.
	The related sub-problem is given by
	\begin{subequations}
		\label{P1-2}
		\begin{eqnarray}
			& \!\!\!\!\!\!\!\!\!\!\!\!\!\!\!\!\! \max  \limits_{ \mathbf{u}} 
			\label{P1-2-function}
			& \!\!\!\!\!\!  \sum \limits_{k \in \mathcal{K}} \widetilde{R}_k  \\
			\label{P1-2-C-RA constraint}  
			& \!\!\!\!\!\!\!\!\!\!\!\!\!\!\!\!\! \mathrm{s.t.} &\!\!\!\!\!\! 	\mathbf{n}(\mathbf{u}_{m,n}^b)^{\mathrm{T}}(\mathbf{q}_j-\mathbf{q}_b) \le 0,  \forall b \in \mathcal{B}, \forall j \in \mathcal{B}, j \neq b,\\
			\label{P1-2-C-RA signal leakage} 
			&&\!\!\!\!\!\!
			\mathbf{n}(\mathbf{u}_{m,n}^b)^{\mathrm{T}}\mathbf{q}_b \ge 0,\forall b \in \mathcal{B}.
		\end{eqnarray}
	\end{subequations}
	Similarly, to tackle constraints $(\mathrm{\ref{P1-2-C-RA constraint}})$ and $(\mathrm{\ref{P1-2-C-RA signal leakage}})$, we have 
	\begin{align}
		\label{P1-2-Constraint1}
		&\bar{\mathbf{n}}^{\mathrm{T}}\mathbf{R}(\Delta\mathbf{u}_{m,n}^b)^{\mathrm{T}}\mathbf{R}(\mathbf{u}_{m,n}^{b,(t-1)})^{\mathrm{T}}(\mathbf{q}_j - \mathbf{q}_b) \le 0, \forall m, \forall n,\\
		\label{P1-2-Constraint2}
		& \bar{\mathbf{n}}^{\mathrm{T}}\mathbf{R}(\Delta\mathbf{u}_{m,n}^b)^{\mathrm{T}}\mathbf{R}(\mathbf{u}_{m,n}^{b,(t-1)})^{\mathrm{T}}\mathbf{q}_b \ge 0, \forall m, \forall n.
	\end{align}
	Let $\mathbf{u}^{(t)} = \mathbf{u}^{(t-1)} + \iota^{(t-1)}(\bar{\mathbf{u}}^{(t-1)} - \mathbf{u}^{(t-1)})$,
	where $\bar{\mathbf{u}}_{m,n}^{(t-1)}$ is a feasible vector and can be chosen as the solution to the following optimization problem:
	\begin{subequations}
		\label{P1-3}
		\begin{eqnarray}
			& \max  \limits_{\mathbf{u}} 
			\label{P1-3-function}
			& \nabla_{\mathbf{u}} C_2(\mathbf{u}^{(t-1)})^{\mathrm{T}} (\mathbf{u} - \mathbf{u}^{(t-1)}) \\
			\label{P1-3-C-RA constraint} 
			& \mathrm{s.t.} & (\mathrm{\ref{P1-2-Constraint1}}), (\mathrm{\ref{P1-2-Constraint2}}),
		\end{eqnarray}
	\end{subequations}
	where $C_2 = \sum \nolimits_{k \in \mathcal{K}} \widetilde{R}_k$. The gradient of function $C_2(\mathbf{u}^{(t-1)})$ at the point $\mathbf{u}^{(t-1)}$ is given by
	\begin{align}
		&\!\!\!\![\nabla_{\mathbf{u}} C_2(\mathbf{u}^{(t-1)})]_j \nonumber \\
		&\!\!\!\!=\! \lim_{\varepsilon \rightarrow 0 }\frac{C_2(\mathbf{u}^{(t-1)} \!+\! \varepsilon\mathbf{e}^j)\!-\! C_2(\mathbf{u}^{(t-1)})}{\varepsilon},1 \le j \le M+N.
	\end{align}
	Problem $(\mathrm{\ref{P1-3}})$ is a linear optimization problem, which can be solved by using linprog.
	
	\vspace{-2mm}
	\subsection{Convergence and Complexity Analysis}
	The overall algorithm for solving problem (\ref{P0}) is presented in Algorithm \ref{Algorithm2}, where the receive beamforming $\mathbf{w}_k$ and rotation angle $\mathbf{u}$ are alternately optimized until convergence is achieved. The objective value of (\ref{P0}) is non-decreasing under the alternating updates of $\mathbf{w}_k$ and $\mathbf{u}$. Moreover, since the original problem (\ref{P0}) is upper-bounded by a finite value, the convergence of the proposed AO algorithm is ensured. For the computational complexity of the proposed algorithm, the complexity of the MMSE-based beamforming scheme is dominated by the matrix inversion, which is $\mathcal{O}(KQ^3)$. For the rotation angles optimization method, the computation of the gradients in $(\ref{u-gradient})$ has a complexity of $\mathcal{O}(K^2Q)$. Thus, the overall complexity is given by $\mathcal{O}(KQ(Q^2+K))$. The convergence and complexity analysis of the proposed algorithm for panel rotation optimization is similar to the antenna rotation optimization, which is omitted here.
	
	\begin{algorithm}[t]  
		\caption{AO Algorithm for Solving Problem (\ref{P0})}
		\label{Algorithm2}
		\renewcommand{\algorithmicrequire}{\textbf{Initialize}}
		\renewcommand{\algorithmicensure}{\textbf{Output}}
		\begin{algorithmic}[1]
			\STATE \textbf{Initialize} $\mathbf{w}_k^{(0)}$ and $\mathbf{u}^{(0)}$, the maximum number of iteration $T_{\max}$, and the predefined threshold $\epsilon_1$.
			\REPEAT 
			\STATE Set the iteration index $\tau=0$;
			\STATE Given $\mathbf{u}^{(\tau)}$, update $\mathbf{w}_k^{(\tau + 1)}$ by calculating $(\mathrm{\ref{W-MMSE}})$;
			\STATE Given $\mathbf{w}_k^{(\tau + 1)}$, update $\{\mathbf{u}^{(\tau)}$ by solving problem (\ref{P3});
			\STATE Update $\tau=\tau+1$;
			\UNTIL $\vert C_1^{(\tau)} - C_1^{(\tau-1)} \vert \le \epsilon_1$ or $\tau \ge T_{\max}$;
			\STATE \textbf{Output} the converged solutions $\mathbf{w}_k$ and $\mathbf{u}$.
		\end{algorithmic}
	\end{algorithm}	
	
	\vspace{-3mm}
	\section{Discrete rotation angels}
	In practical implementation, continuous rotation angles adjustment is challenging due to the limited flexibility of the mechanical- and electronical-components. This motivates us to consider the discrete rotation angels design in this section. 
	Specifically, we assume that the discrete rotations sets for $\alpha$ and $\beta$ are given by $\mathcal{L}_{\alpha}=[1,2,\cdots,L_{\alpha}]$ and $\mathcal{L}_{\beta}=[1,2,\cdots,L_{\beta}]$, respectively. Then, the $(l_1,l_2)$-th discrete rotation at the $m$-th row and $n$-th column of the BS, i.e., antenna $a_{m,n}$, is given by $\mathbf{u}_{m,n}^l=[\alpha_{m}^{l_1},\beta_{n}^{l_2}]$, where $l_1 \in \mathcal{L}_{\alpha}$ and $l_2 \in \mathcal{L}_{\beta}$. 
	The related discrete rotation angle for antenna rotation optimization problem is given by
	\begin{subequations}
		\label{P5-1}
		\begin{eqnarray}
			& \!\!\!\!\!\!\!\!\!\!\!\!\!\!\!\!\! \max  \limits_{\mathbf{w}_k, \mathbf{u}} 
			\label{P5-1-function}
			& \!\!\!  \sum \limits_{k \in \mathcal{K}} R_k  \\
			\label{P5-1-C-Discrete constraint} 
			&\!\!\!\!\!\!\!\!\!\!\!\!\!\!\!\!\! \mathrm{s.t.} &\!\!\! (\mathrm{\ref{P0-C-receive power} }), (\mathrm{\ref{P0-C-RA constraint} }), \alpha_m \in \mathcal{L}_{\alpha}, \beta_n \in \mathcal{L}_{\beta}, \forall m, \forall n.
		\end{eqnarray}
	\end{subequations}
	The receive beamforming $\mathbf{w}_k$ can be similarly obtained by MMSE method. 
	However, since the variable $\mathbf{u}_{m,n}$ is discrete and coupled, problem (\ref{P5-1}) is still difficult to solve. Considering that the rotation angles $\alpha_m$ and $\beta_n$ are strongly coupled in determining the pointing vector of each antenna $a_{m,n}$, and their relationship with the antenna pointing vector is nonlinear. Simply rounding each continuous angle to the nearest discrete value can lead to a joint beamforming direction that deviates significantly from the optimum. Moreover, the resulting configuration may fail to satisfy the mechanical rotation constraints specified in $(\mathrm{\ref{P0-C-RA constraint} })$.
	To address this, we adopt a discrete genetic algorithm (GA), which is designed to directly search for optimal solutions in discrete parameter spaces. This enables effective and constrained optimization of the antenna rotation angles in the following process.
	
	The main steps of the discrete genetic algorithm include initialization, selection, crossover, and mutation.
	
	\begin{itemize}
		\item \textbf{Initialization:} 
		A population of potential solutions is created randomly. Specifically, the chromosome consists of two parts: the first $M$ genes represent discrete indices for $\alpha_m$ (selected from set $\mathcal{L}_{\alpha}$), and the next $N$ genes represent discrete indices for $\beta_n$ (selected from set $\mathcal{L}_{\beta}$).
		
		\item \textbf{Selection:} 
		Candidate solutions are evaluated based on a fitness function $F(\mathbf{x})$. For each chromosome, the discrete values of $\alpha_m$ and $\beta_n$ are first decoded, then the geometric constraints are checked. The fitness function measures the degree of constraint satisfaction, prioritizing solutions that satisfy more constraints. The best-performing solutions are selected as parents for the next generation using tournament selection.
		
		\item \textbf{Crossover:} 
		Due to the discrete nature of the problem, a two-point crossover method is employed: two crossover points are randomly selected, and the gene segments between these points are exchanged between parent chromosomes. Crossover is performed with probability $p_c$; otherwise, the parents are copied directly.
		
		\item \textbf{Mutation:} 
		Mutation operates on each gene with probability $p_m$: for genes representing $\alpha_m$, a new index is randomly selected from $\mathcal{L}_{\alpha}$ (different from the current value); for genes representing $\beta_n$, a new index is randomly selected from $\mathcal{L}_{\beta}$. This ensures that the mutated angle values remain within the discrete sets.
	
	\end{itemize}
	
	Let $P$, $G_{\max}$, and $\eta$ denote the population size, the maximum number of generations, and the selection pressure, respectively.
	The algorithm iterates through these steps until the termination criterion is met.
	The fitness function $F(\mathbf{x})$ is expressed as
	\begin{eqnarray}
		\! F(\mathbf{x}) \!= \!
		\begin{cases}
			C(\mathbf{x}), \!\!& \text{if rotation constraint $(\mathrm{\ref{P0-C-RA constraint}})$ is satisfied} , \\
			\xi N_v, \!\!& \text{otherwise},
		\end{cases}
	\end{eqnarray}
	where $\mathbf{x}=[\alpha_1,\cdots, \alpha_M,\beta_1,\cdots, \beta_N]$ represents the discrete rotation angles in the $g$-th generation. Besides, $C(\mathbf{x}) = \sum \nolimits_{k \in \mathcal{K}} R_k$. Here $\xi < 0$ and $N_v$ denote the penalty coefficients and the number of violations of the constraints, respectively. For the optimization of the discrete angles of the panel rotation, we can apply the similar algorithm, which is omitted here. 
	This genetic algorithm efficiently explores the discrete solution space of size $L^{M+N}$, leveraging the parallel evaluation of multiple candidate solutions to find optimal or near-optimal configurations of $\alpha_m$ and $\beta_n$ angles that satisfy the geometric constraints while respecting the mechanical limitations of the system.

	\vspace{-4mm}
	\section{Simulation Results} \label{Simulation Results}
	\vspace{-1mm}
	\begin{table*}[t]
		\centering
		\caption{Simulation Parameters}
		\scalebox{0.9}{\begin{tabular}{c c} 
				\toprule
				\bf{Parameter}   & \bf{Value}  \\ 
				\midrule
				Number of users, antennas, and scatters   & $K=6$, $Q=64$, and $D=8$\\
				Inter-antenna spacing   & $\Delta = \lambda / 2$ \\
				Number of vertical and horizontal rotating tracks for antenna optimization & $M = 8$ and $N=8$\\
				the array occupation ratio and the antenna directivity factor & $\zeta = 1$ and $p=2$ \\
				Maximum zenith angle & $\theta_{\max} = \pi/6$ \\
				Number of blocks and the antennas at each block & $B=4$ and $Q_b=16$ \\
				Number of vertical and horizontal rotating tracks for panel optimization & $M_B = 2$ and $N_B=2$ \\
				Noise power at the BS and the maximum transmit power of each user& $\sigma^2 = -80~\mathrm{dBm}$ and $P_{\max}=10~\mathrm{dBm}$ \\
				Genetic algorithm parameters & $P = 200$, $G_{\max}=100$, $p_c = 0.8$, $p_m=0.1$, and $\eta = 2$ \\
				\bottomrule
		\end{tabular}}
		\label{Para}
		\vspace{-3mm}
	\end{table*}
	
	In this section, we present simulation results to evaluate the performance of the CL RA-enabled system. We assume that the system operates at a frequency of $3.5$ GHz, corresponding to a wavelength of $\lambda =0.0857$ m. We consider a global Cartesian coordinate system with the origin at the BS. Users are distributed within a circular area with a distance of $50 \sim 70$ meters from the BS reference center.
	Other parameters are given in Table \ref{Para}.
	To evaluate the system performance, the following schemes are considered:
	\begin{figure}[t]
		\centering
		\includegraphics[width=0.41
		\textwidth]{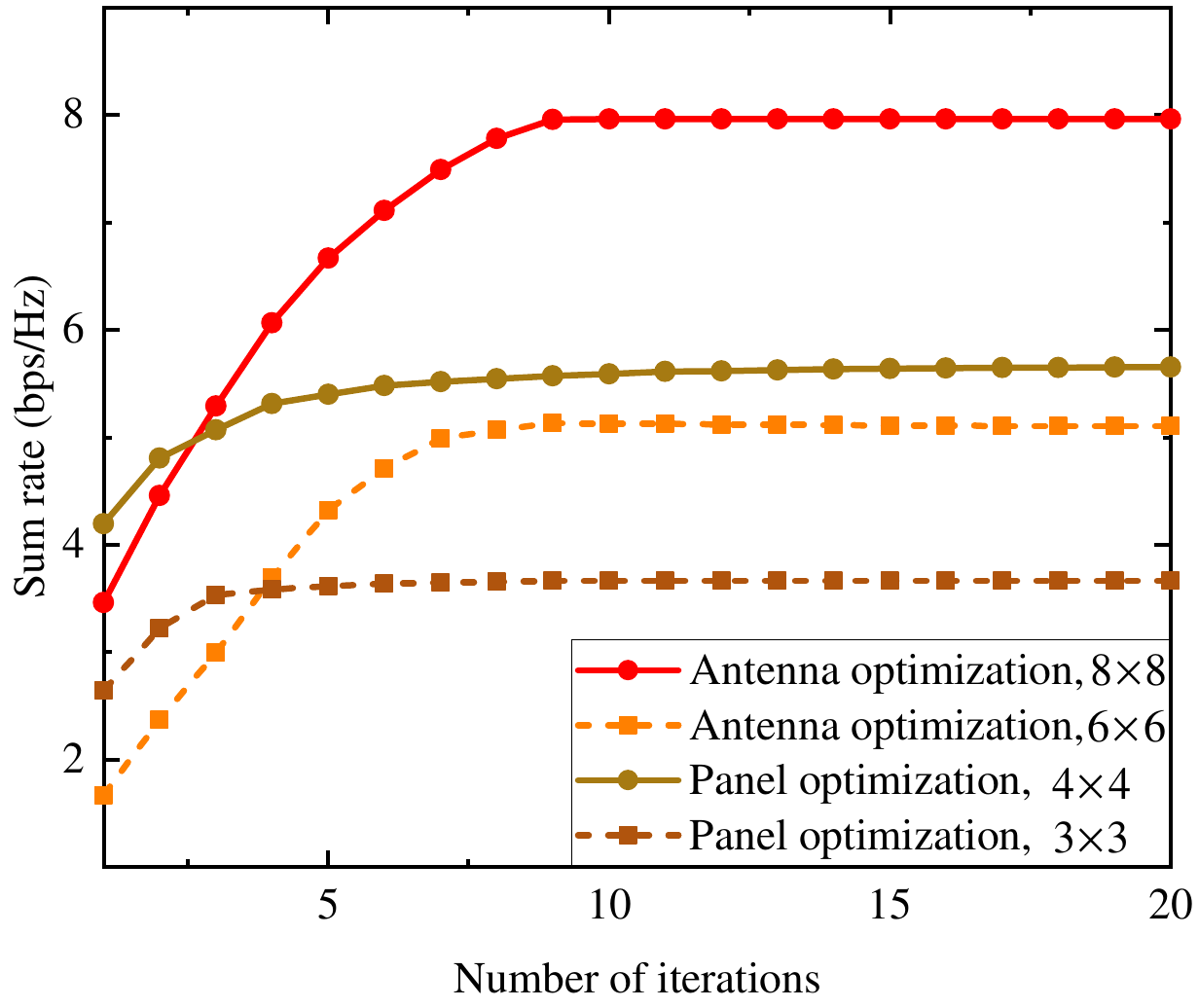}
		\caption{Convergence behavior of the proposed algorithms.}
		\label{Fig2}
		\vspace{-5mm}
	\end{figure}
	
	\begin{itemize}
		\item \textbf{Flexible antenna orientation}: The orientation of each RA can be adjusted independently, and the MMSE method is applied at the BS.
		\item \textbf{Flexible panel orientation}:  The rotation angel of each panel can be adjusted independently. The BS utilizes the MMSE method for receive beamforming.
		\item \textbf{Random orientation design}: The orientation of each RA is randomly generated within the rotation ranges given by $(\mathrm{\ref{theta_constarint}})$, and the MMSE method is applied at the BS.
		\item
		\textbf{Array-wise orientation design}:  
		We adjust the orientation of the entire antenna array instead of that of each antenna element by setting $B=1$, and perform MMSE at the BS.
		\item 
		\textbf{Fixed orientation design}: The orientations of all RAs are fixed at their reference orientations, i.e., $\mathbf{f}_{m,n}=\mathbf{e}_x$. The MMSE method is applied at the BS.
		\item \textbf{Baseline with isotropic antennas}: The antenna elements in the array are isotropic, i.e., $p=0$, and the radiation energy is evenly distributed in the front half-space of the antennas. The MMSE method is applied at the BS.
	\end{itemize}
	
	\vspace{-3mm}
	\subsection{Convergence of Proposed Algorithm}
	Fig. \ref{Fig2} illustrates the convergence behavior of the proposed algorithms, i.e., antenna optimization and panel optimization. In the panel schemes, each array consists of $B=2 \times 2$ panels. In the $4 \times 4$ panel scheme, each panel contains $4 \times 4$ antennas, where $Q=64$. In the $3 \times 3$ panel scheme, each panel contains $3 \times 3$ antennas, where $Q=36$. It is observed that the sum rate increases over iterations and converges within $10$ iterations in all cases. This indicates that the proposed AO algorithms converge quickly and demonstrate their effectiveness. Moreover, a larger array results in a greater sum rate owing to its larger aperture and higher directivity. 
	In general, antenna optimization schemes outperform their antenna panel-level counterparts. The key advantage lies in their utilization of a higher number of DoFs, which allows for dynamic and precise beamforming adjustments.

	\vspace{-3mm}
	\subsection{System Throughput Versus Transmit Power of Each User}
	\begin{figure}[t]
		\centering
		\includegraphics[width=0.41
		\textwidth]{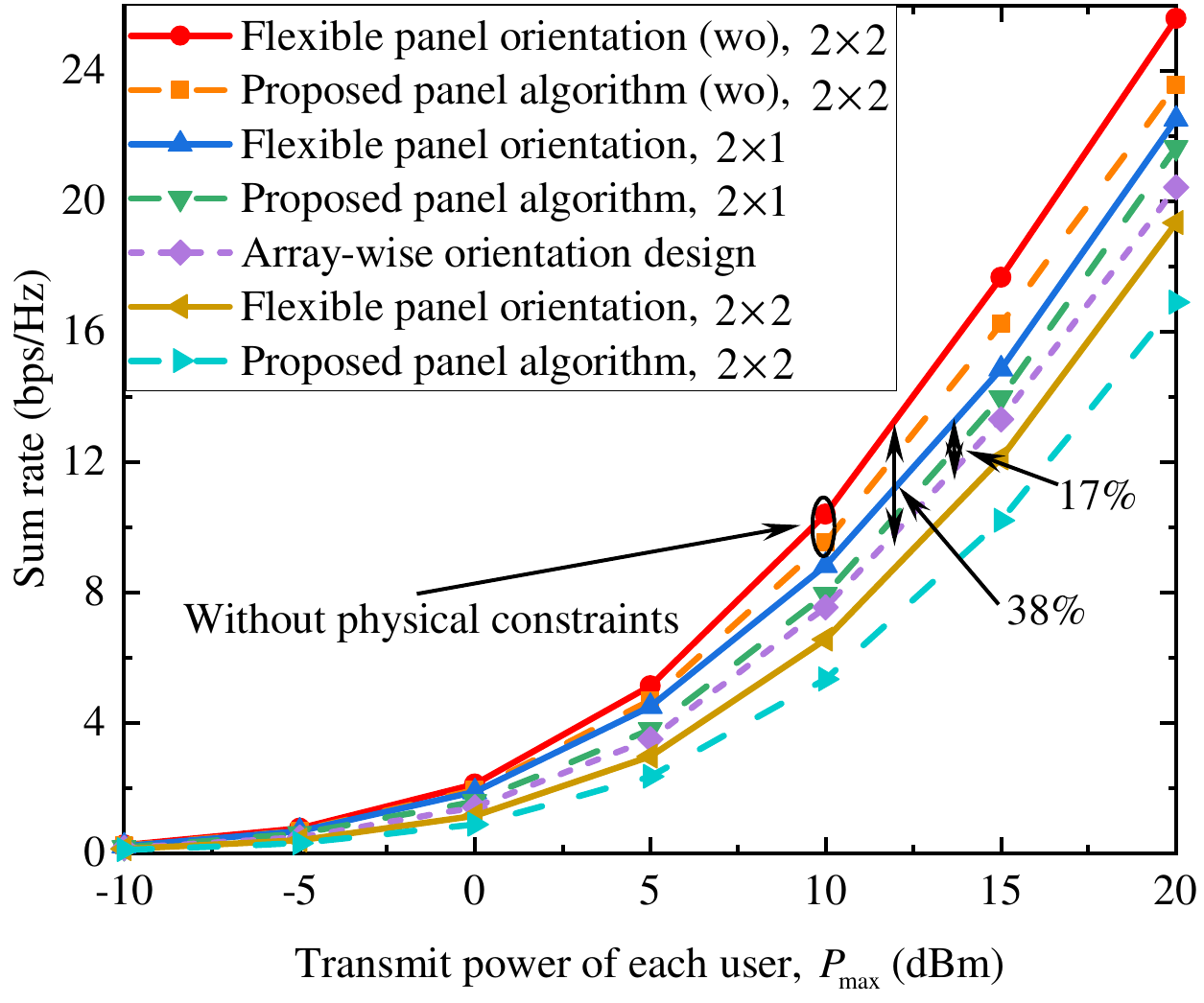}
		\caption{Sum rate vs. the maximum transmit power of each user.}
		\label{Fig3-1}
		\vspace{-5mm}
	\end{figure}
	In Fig. \ref{Fig3-1}, the flexible panel orientation (wo) and the proposed panel algorithm (wo) schemes indicate that panel orientation is optimized without considering practical physical constraints. We present different panels partition schemes, i.e., $2 \times 1$ and $2 \times 2$, where $Q=64$. It is observed that without any physical constraints, the performance of flexible panel orientation and proposed panel algorithm is superior to that of the array-wise orientation design scheme. This is reasonable because the panel partition provides larger DoFs, where the rotation angle of each panel can be independently adjusted to align with the desired direction. Moreover, although the CL structure restricts some rotational DoFs, it delivers a favorable performance-complexity trade-off, leading to an overall performance improvement. Specifically, the flexible panel orientation (wo) shows a performance improvement of 38\% compared to the array-wise orientation design scheme when $P_{\mathrm{max}}=10~\mathrm{dBm}$. 
	Despite strict physical constraints that limit the rotational DoFs and thereby reduce the feasible range for optimization, the $2 \times 1$ panel rotation scheme still achieves a $17\%$ performance gain compared to the array-wise orientation design scheme. 
	However, as the number of panel increases, the performance decreases. Particularly, the $2 \times 2$ panel rotation scheme is inferior to the array-wise orientation design scheme.

	\begin{figure}[t]
		\centering
		\includegraphics[width=0.41
		\textwidth]{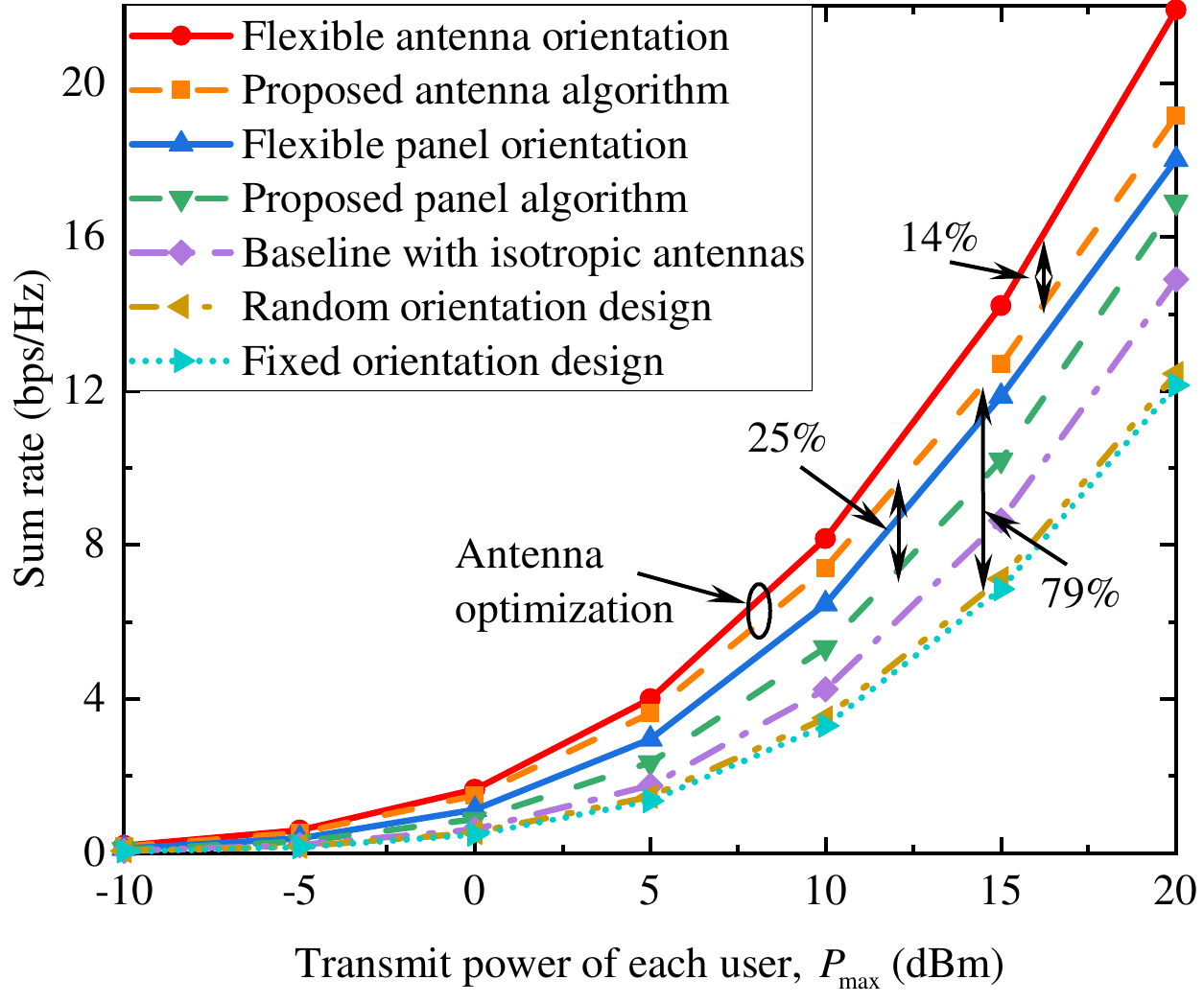}
		\caption{Sum rate vs. the maximum transmit power of each user.}
		\label{Fig3}
		\vspace{-5mm}
	\end{figure}
	In Fig. \ref{Fig3}, we compare the proposed antenna/panel algorithm with different baselines. Obviously, the sum rate of all schemes increases as the transmit power $P_{\max}$ increases. Besides, the flexible antenna orientation scheme consistently achieves the highest sum rate across the all power levels. As $P_{\max}$ increases, the performance gap between the flexible antenna orientation and the proposed antenna algorithm widens. This is because the increase of $P_{\max}$ leads to the interference-limited operation and the flexible antenna orientation support more flexibly adjustments to steer main lobes toward intended users. Particularly, the performance gap is $14\%$ when $P_{\max}=20~\mathrm{dBm}$. Moreover, the proposed antenna algorithm is superior to the proposed panel algorithm. This is primarily because the proposed antenna algorithm employs element-wise rotation, though binding per row and column, enabling finer beam control.
	When $P = 20~\mathrm{dBm}$, the performance gain of the proposed antenna algorithm surpasses that of the proposed panel algorithm by $25\%$.
	The fixed and random orientation designs exhibit the lowest performance, highlighting the limitation of static or unoptimized configurations. Specifically, at $P_{\max}=15~\mathrm{dBm}$, the proposed antenna algorithm achieves a performance gain $79\%$ higher than that of the random orientation design.

	\vspace{-3mm}
	\subsection{System Throughput Versus Maximum Zenith Angle}
	\begin{figure}[t]
		\centering
		\includegraphics[width=0.41
		\textwidth]{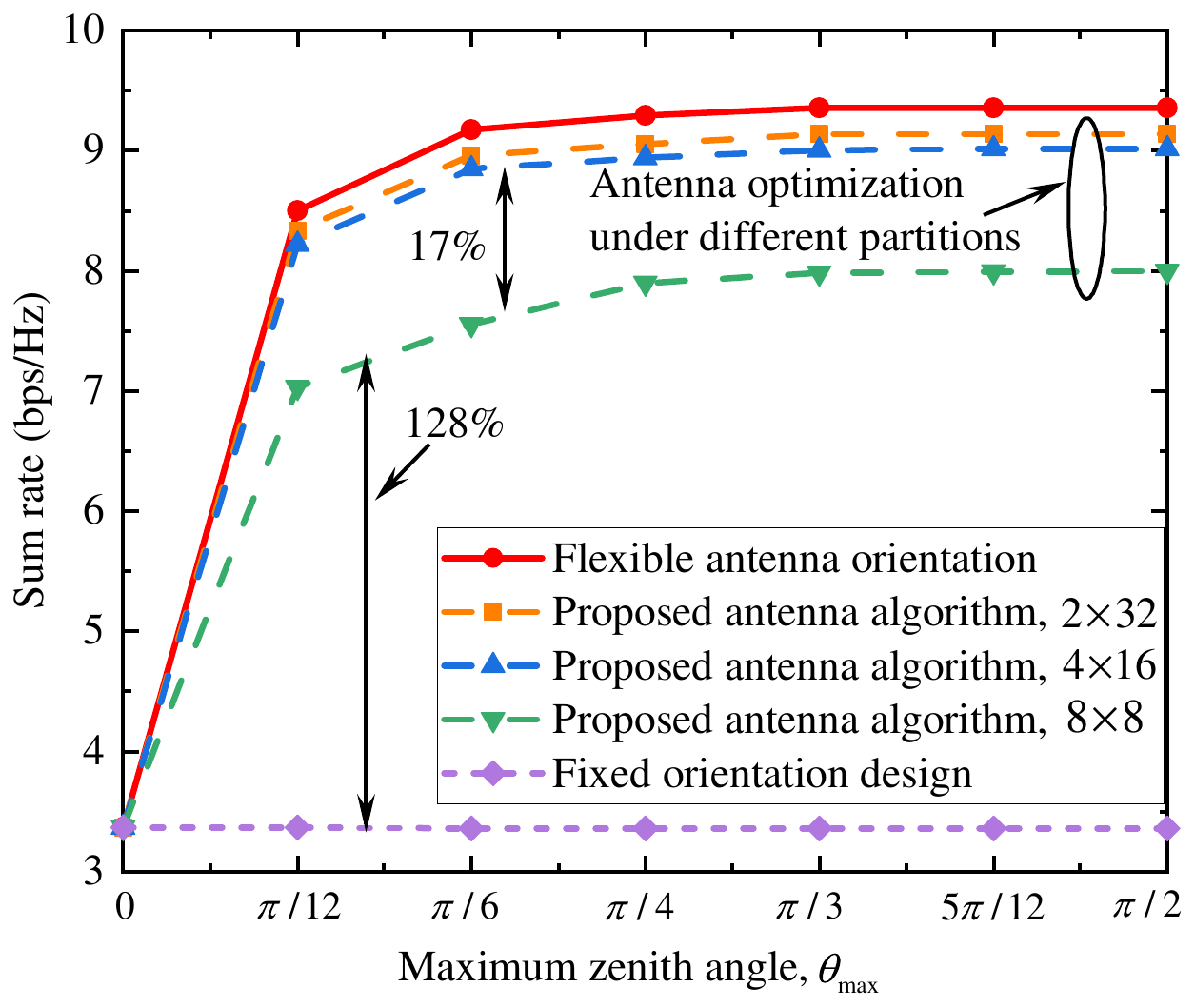}
		\caption{Sum rate vs. the maximum zenith angle.}
		\label{Fig5}
		\vspace{-5mm}
	\end{figure}
	
	In Fig. \ref{Fig5}, it is no doubt that the flexible antenna orientation exceeds all counterparts around the entire range. When $\theta_{\max} = 0$, the RA schemes yield equivalent performance to the fixed orientation design. Surprisingly, a slight increase of $\theta_{\max}$ to $\pi/12$ leads to a dramatic improvement in the performance of RA schemes, where the proposed antenna algorithm, equipped with an $8 \times 8$ array, achieves a $128\%$ higher performance gain than the fixed orientation design at this point. This result indicates that despite a limited rotation range for RA orientation/boresight adjustment, the proposed RA system can still achieve significant performance gains through optimized pointing directions. Moreover, it can be observed that the performance of all RA schemes tends to saturate at  $\theta_{\max} = \pi/4$. 
	
	Across different array configurations, it can be observed that the $2 \times 32$ array achieves the highest performance, with the $4 \times 16$ array closely following at a minimal gap. The $8  \times 8$ array exhibits the worst performance, showing a performance deficit of $17\%$ compared to the $4 \times 16$ array. This can be explained as follows. Compared to the two counterparts, the $2 \times 32$ array gains more DoFs and flexibility to balance the array directional gain over the multipath channels, thus bringing a further performance increase. However, with more motors to support rotation, the hardware cost rises consistently. The performance gap between the flexible antenna orientation and $4 \times 16$ array is small, with a mere $0.2~\mathrm{bps/Hz}$ difference. This reveals that with an appropriate array partitioning strategy, the proposed antenna algorithm can achieve performance comparable to that of the flexible antenna orientation scheme, while simultaneously reducing deployment complexity and energy consumption.
	
	\vspace{-3mm}
	\subsection{System Throughput Versus Directivity factor}
	\begin{figure}[t]
		\centering
		\includegraphics[width=0.41
		\textwidth]{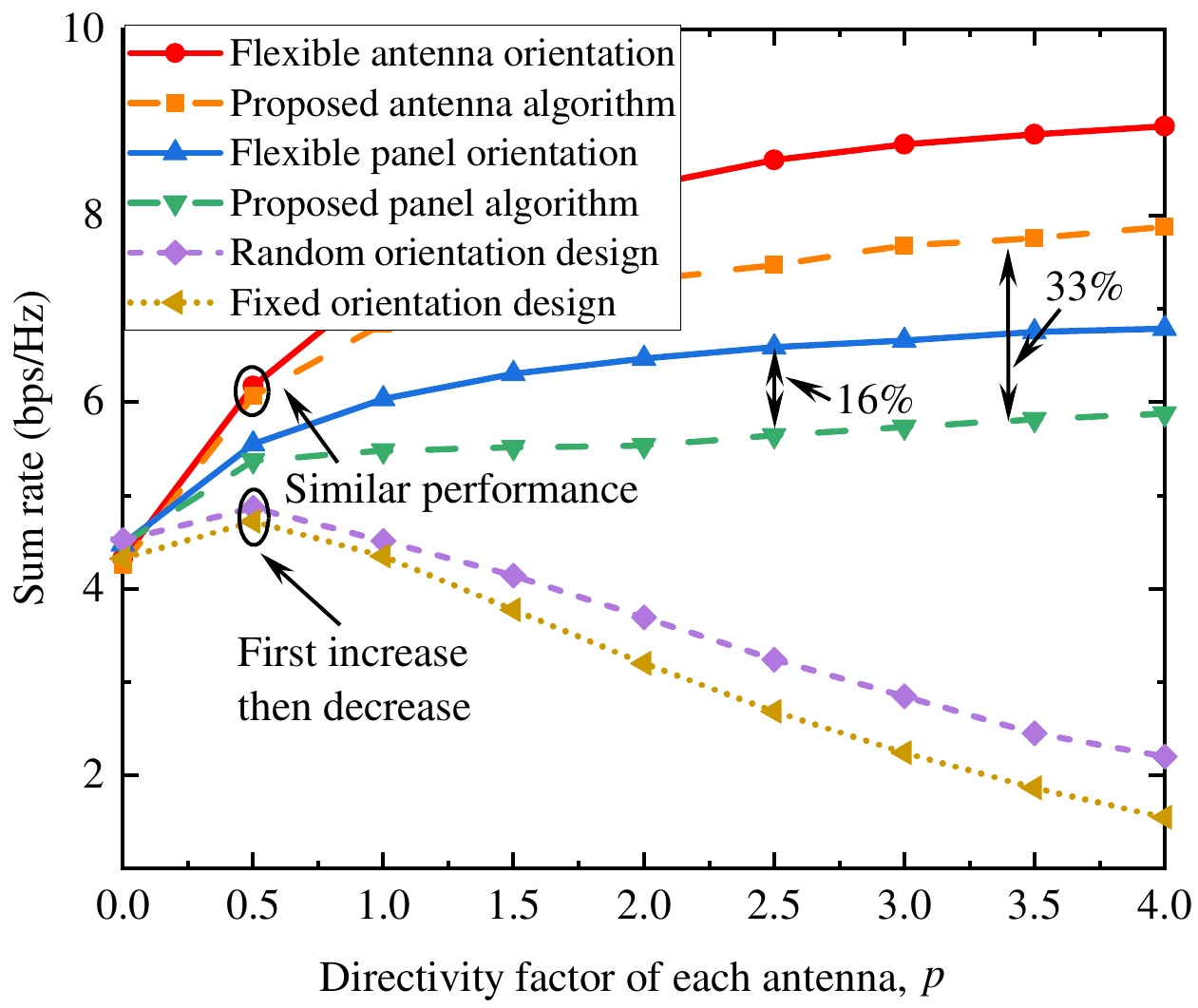}
		\caption{Sum rate vs. the directivity factor of each antenna.}
		\label{Fig6}
		\vspace{-5mm}
	\end{figure}
	
	As observed in Fig. \ref{Fig6}, the performance of all RA schemes improves as the directivity factor $p$ increases. It is reasonable that a larger $p$ leads to a higher directional gain in the antenna's boresight direction and a narrower main lobe towards intended users. The performance gap between the flexible antenna orientation and the proposed antenna algorithm widens with the increase of $p$. This is because the former can enhance directional gain for multiple users by exploiting the narrow  main lobe and reducing mutual interference, thus achieving higher system performance. However, constrained by the limited DoFs caused by the bound row and column rotation, the proposed antenna algorithm struggles to achieve precise beam alignment for multiple users. This limitation consequently leads to slower performance growth.
	Similarly, when $p \le 0.5$, the proposed panel algorithm exhibit the same performance as the flexible panel orientation.
	However, with the increase of $p$, the performance gain brought by the higher rotational DoFs occurs. Specifically, compared to the proposed panel algorithm, the proposed antenna algorithm and flexible panel orientation exhibit sum rate gains of $33\%$ and $16\%$, respectively.
	
	In addition, the fixed orientation design decreases with $p$ when $p \ge 0.5$.
	This is because as the directivity factor $p$ increases, the radiation power becomes more concentrated in the region directly in front of the array. 
	Consequently, users deviating from the main lobe experience weaker directional gains.
	Furthermore, the random orientation design is outperformed by the proposed RA schemes. This is because despite its ability to disperse the array's radiation power in multiple directions, it lacks the strategic resource allocation to improve communication performance for all users. The above results demonstrate the necessity of the proposed RA system for enhancing channel capacity, particularly critical when dealing with highly directive antennas with narrow main lobes.
	
	\vspace{-3mm}
	\subsection{Continuous Versus Discrete Angle Rotation}
	\begin{figure}[t]
		\centering
		\includegraphics[width=0.41
		\textwidth]{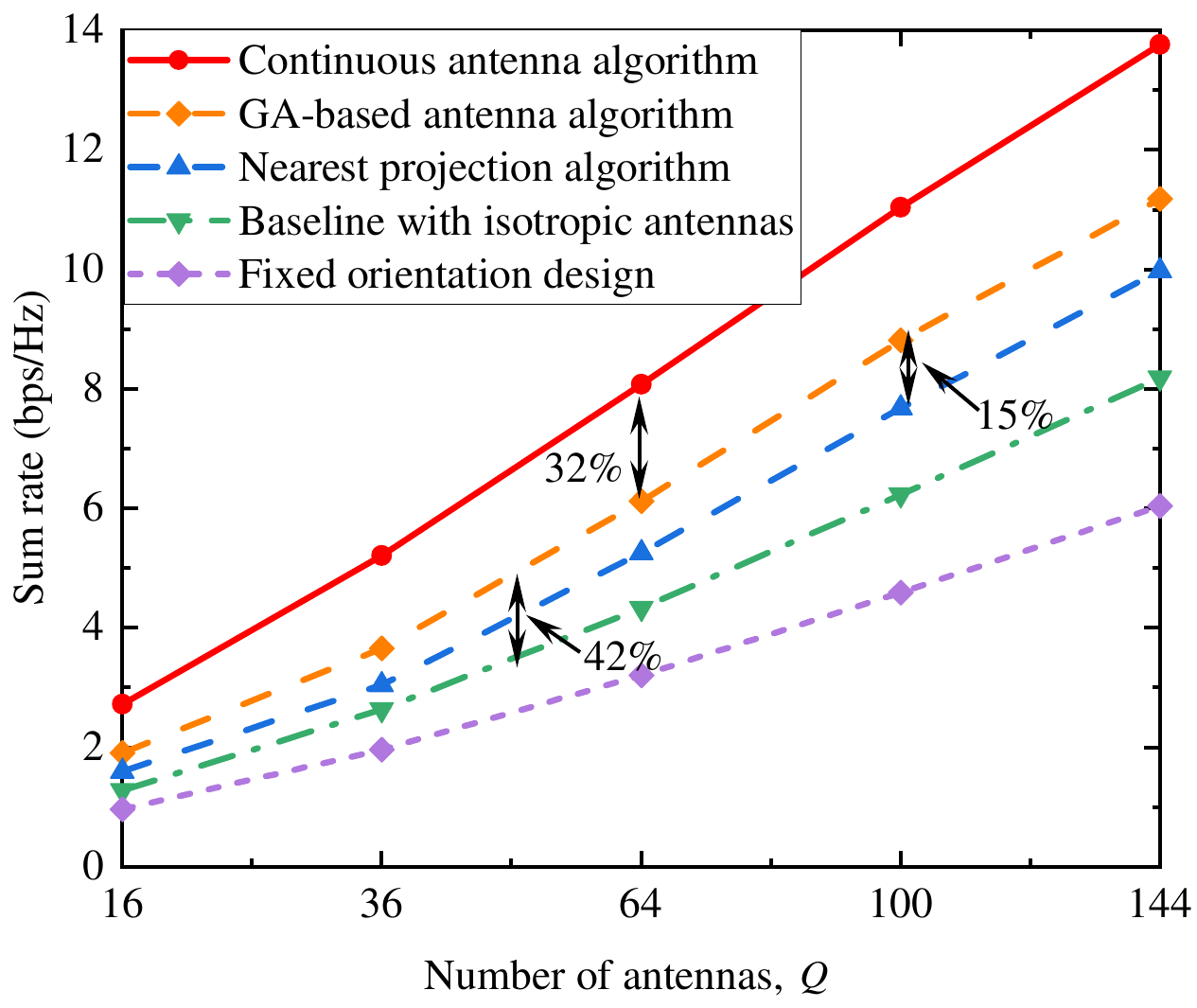}
		\caption{Sum rate vs. the number of antennas.}
		\label{Fig4}
		\vspace{-5mm}
	\end{figure}

	Fig. \ref{Fig4} shows the performance of different discrete rotation angle selection algorithms with $L=L_{\alpha}=L_{\beta} = 15$. 
	Specifically, the nearest projection algorithm is defined as mapping the continuous optimum to its closest feasible discrete point while adhering to the mechanical rotation limits.
	As expected, increasing $Q$ improves the sum rate due to higher array gains.
	The continuous antenna algorithm outperforms the GA-based antenna algorithm by 32\% at $Q=64$. This advantage stems from the fact that the continuous solution space is inherently larger and less constrained than its discrete counterpart, enabling more precise and globally optimal configurations.
	Moreover, compared with the nearest projection algorithm, the GA-based antenna algorithm achieves a 15\%  performance gain when $Q=100$. This demonstrates the effectiveness of the GA in searching the discrete space.  
	Unlike nearest projection, which is confined to a local neighborhood and ignores variable coupling, the GA performs a joint global optimization over all orientations. This allows it to escape poor local optima and find a better-coordinated discrete solution that nearest projection cannot reach.
	Furthermore, the baseline with isotropic antennas performs consistently worse across the entire range, underscoring the importance of directional gain and intelligent orientation control in achieving higher spectral efficiency. Notably, the GA-based antenna algorithm achieves a $42\%$ improvement over this baseline.

	\begin{figure}[t]
		\centering
		\includegraphics[width=0.41
		\textwidth]{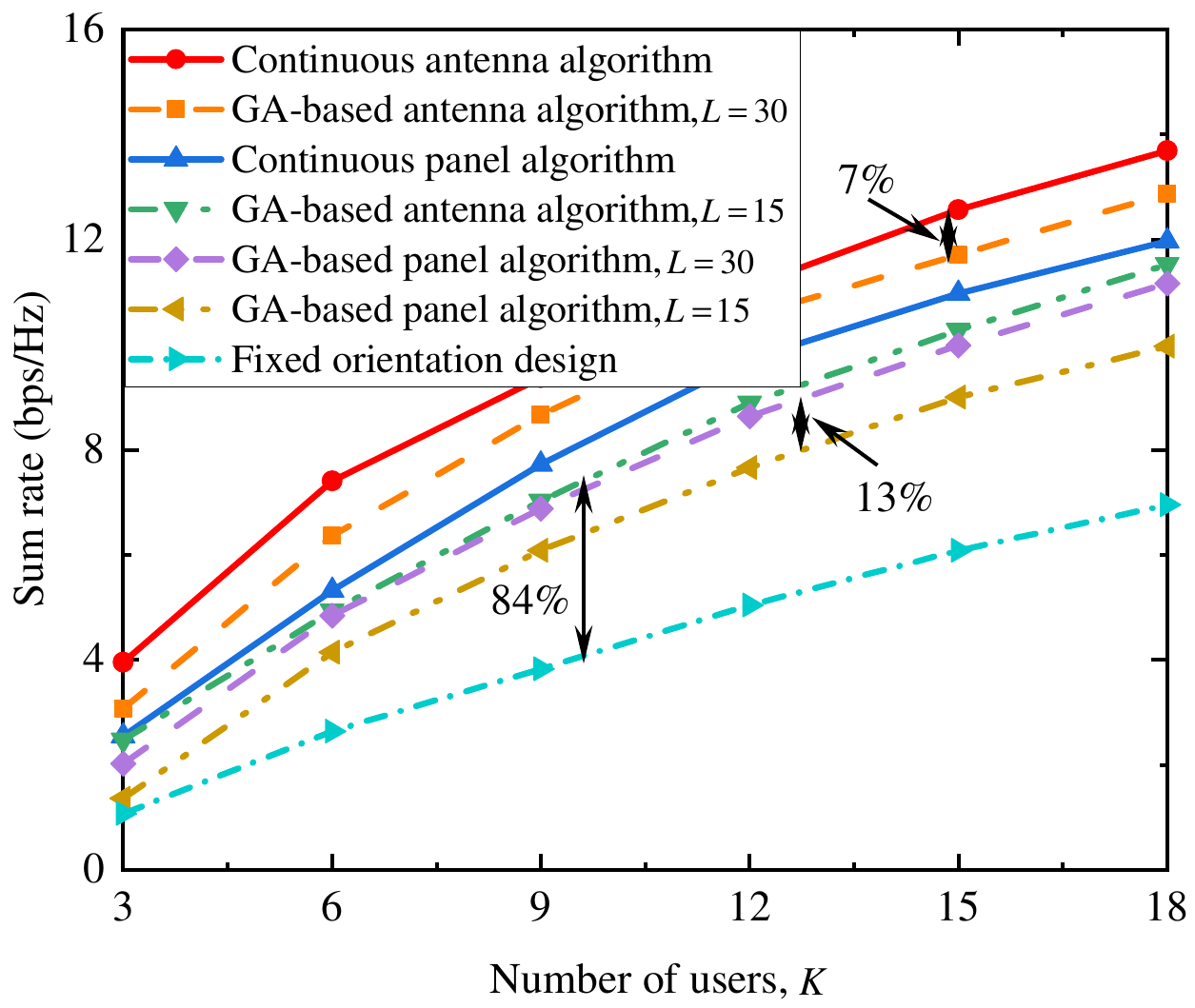}
		\caption{Sum rate vs. the number of users.}
		\label{Fig7}
		\vspace{-5mm}
	\end{figure}
	
	In Fig. \ref{Fig7}, we evaluate the sum rate versus the number of users $K$ under different discrete rotation configurations, where $L=L_{\alpha}=L_{\beta}$. As expected, the sum rate grows with $K$ across all schemes. Continuous rotation algorithms consistently outperform their discrete counterparts. For instance, the continuous antenna algorithm achieves a $7\%$ higher sum rate than the GA-based antenna algorithm with $L=30$.
	Moreover, the performance of the GA-based antenna/panel algorithm increases as $L$ increases. This is reasonable because 
	a larger rotation angle set will provide more spatial DoFs for tuning the antenna/panel into the desired directions, thereby enhancing the array gains. Specifically, the GA-based panel algorithm with $L=30$ outperforms the version with $L=15$ by 13\%.
	Furthermore, even with a coarse discretization ($L=15$), the GA-based antenna algorithm still surpasses the fixed orientation design by $84\%$. Additionally, this performance gap widens as $K$ increases, highlighting the system's ability to support multi-user scenarios through adaptive beam alignment. These results underscore the effectiveness of RA-based systems in dense user environments.

	\vspace{-2mm}
	\section{Conclusions} \label{Conclusions}
	In this paper, we investigated a CL-RA architecture to reduce the hardware costs and control complexity related to the conventional RA systems. 
	We considered a CL-RA-aided uplink multi-user system, establishing system models for both antenna element-level and antenna panel-level rotation schemes. To maximize the system throughput, we formulated a joint optimization problem for the receive beamforming and rotation angles.
	We first revealed that the CL-RA structure can achieve the same performance as the flexible rotation scheme for the antenna element-level rotation in a single-user scenario.
	Subsequently, for the multi-user case under both rotation architectures, we developed an AO algorithm to address the formulated problem, where the MMSE and feasible direction methods were exploited to obtain the receive beamforming and rotation angles, respectively. Additionally, we addressed the practical discrete angle selection problem via a genetic algorithm. Simulation results validated the effectiveness of the CL-RA architecture. Despite the limited rotational DoFs, by designing the row-column partition, the proposed algorithm can achieve performance similar to that of the flexible orientation scheme with significantly reduced hardware costs. Moreover, due to strict physical constraints,
	increasing the number of panels reduced their rotational DoFs, ultimately degrading performance.
	In addition, the performance of the proposed antenna scheme was 25\% and 128\% higher than that of the proposed panel scheme and the fixed orientation design, respectively.

	\appendices
	\section*{{Appendix A:} \hspace{0.1cm} {Proof of Proposition \ref{pro2}}}  \label{app2}
	The outward normal vector of the panel $b$ is given by
	\begin{align}
		\mathbf{n}(\mathbf{u}_{m,n}^b) =\mathbf{R} (\mathbf{u}_{m,n}^b)\bar{\mathbf{n}}=[c_{\alpha_m}c_{\beta_n} \ -s_{\beta_n} \ s_{\alpha_m} c_{\beta_n} ]^\mathrm{T}.
	\end{align}
	For physical constraint $\mathbf{n}(\mathbf{u}_{m,n}^b)^\mathrm{T}(\mathbf{q}_j-\mathbf{q}_b) \le 0$, we have  
	\begin{align}
		\mathbf{q}_j-\mathbf{q}_b = [0 \ (n_j-n)\Delta \ (m_j-m)\Delta].
	\end{align}
	Then, we can get
	\begin{align}
		&\mathbf{n}(\mathbf{u}_{m,n}^b)^\mathrm{T}(\mathbf{q}_j-\mathbf{q}_b) \nonumber \\
		&= [c_{\alpha_m}c_{\beta_n} \ -s_{\beta_n} \ s_{\alpha_m} c_{\beta_n} ][0 \ (n_j-n)\Delta \ (m_j-m)\Delta]^\mathrm{T} \nonumber \\
		& = -s_{\beta_n}(n_j-n)\Delta+s_{\alpha_m} c_{\beta_n}(m_j-m)\Delta \le 0.
	\end{align}
	Finally, we obtain
	\begin{align}
		\label{Constr1-1}
		-s_{\beta_n}(n_j-n)+s_{\alpha_m} c_{\beta_n}(m_j-m) \le 0.
	\end{align}
	For panel $b$, it should satisfy constraint $(\ref{Constr1-1})$ for any panel $j$, where $ \forall j \neq b, j \in \mathcal{B}$.
	For constraint $\mathbf{n}(\mathbf{u}_{m,n}^b)^\mathrm{T}\mathbf{q}_b \ge 0$, we transform it as 
	\begin{align}
		-s_{\beta_n}n+s_{\alpha_m} c_{\beta_n}m \ge 0 \rightarrow s_{\alpha_m} c_{\beta_n}m \ge s_{\beta_n}n.
	\end{align}
	
	\vspace{-1mm}
	\section*{{Appendix B:} \hspace{0.1cm} {Proof of Proposition \ref{pro1}}} \label{app1}
	When ${\mathbf{f}}(\mathbf{u}_{m,n}) = \mathbf{u}_0$, we have
	\begin{align}
		&c_{\alpha_m}c_{\beta_n} = \frac{x_0}{r}, 
		-s_{\beta_n} = \frac{y_0-y_{n}}{r},
		s_{\alpha_m} c_{\beta_n} = \frac{z_0-z_{m}}{r},
	\end{align}
	where $r=\sqrt{x_0^2+(y_0-y_n)^2+(z_0-z_m)^2}$.
	Thus, since $\sin(\beta_n)=\frac{y_n-y_0}{r}$, we can obtain
	\begin{align}
		\label{cos_alpha}
		& \cos(\alpha_m) = \frac{x_0}{\sqrt{x_0^2+(z_0-z_m)^2}},\\
		\label{cos_beta}
		&  \cos(\beta_n) = \frac{\sqrt{x_0^2+(z_0-z_m)^2}}{r}.
	\end{align}
	When constraint $(\mathrm{\ref{P0-C-RA constraint}})$ is not taken into account, 
	we can obtain the optimal $\alpha_m$  and $\beta_n$ from $(\mathrm{\ref{cos_alpha}})$ and $(\mathrm{\ref{cos_beta}})$, respectively. However, it is observed that $\beta_n$ is related to $z_m$, which implies that the optimal rotation angle $\beta_n$ for each antenna located in the $n$-th column is not the same. Thus, we consider $M=1$, i.e., $z = z_m$. In this case, each RA can tune its antenna orientation towards the user, thereby achieving  the best performance. 
	
	\vspace{-3mm}
	\begin{spacing}{0.97}
		\bibliographystyle{IEEEtran}
		\bibliography{thesis1}
	\end{spacing}
	
\end{document}